\documentclass[8pt]{amsart}

\usepackage{amsmath}
\usepackage{amssymb}
\usepackage{amsthm}
\usepackage{enumerate}
\usepackage[all,arc,curve,color,frame]{xy}
\usepackage[utf8]{inputenc}

\usepackage{color, graphicx}
\usepackage[square,sort,comma,numbers]{natbib}

\renewcommand\section{\@startsection{section}{1}{\z@}
                                   {-3.5ex \@plus -1ex \@minus -.2ex}
                                   {2.3ex \@plus .2ex}
                                   {\normalfont\large\bfseries}}
\renewcommand\subsection{\@startsection{subsection}{2}{\z@}
                                   {-3.25ex\@plus -1ex \@minus -.2ex}
                                   {1.5ex \@plus .2ex}
                                   {\normalfont\normalsize\bfseries}}
\renewcommand\subsubsection{\@startsection{subsubsection}{3}{\z@}
                                   {-3.25ex\@plus -1ex \@minus -.2ex}
                                   {1.5ex \@plus .2ex}
                                   {\normalfont\normalsize\bfseries}}
\renewcommand\paragraph{\@startsection{paragraph}{4}{\z@}
                                   {3.25ex \@plus1ex \@minus.2ex}
                                   {-1em}
                                   {\normalfont\normalsize\bfseries}}

\makeatother

\newdimen\tableauside\tableauside=1.0ex
\newdimen\tableaurule\tableaurule=0.4pt
\newdimen\tableaustep
\def\phantomhrule#1{\hbox{\vbox to0pt{\hrule height\tableaurule
width#1\vss}}}
\def\phantomvrule#1{\vbox{\hbox to0pt{\vrule width\tableaurule
height#1\hss}}}
\def\sqr{\vbox{%
  \phantomhrule\tableaustep

\hbox{\phantomvrule\tableaustep\kern\tableaustep\phantomvrule\tableaustep}%
  \hbox{\vbox{\phantomhrule\tableauside}\kern-\tableaurule}}}
\def\squares#1{\hbox{\count0=#1\noindent\loop\sqr
  \advance\count0 by-1 \ifnum\count0>0\repeat}}
\def\tableau#1{\vcenter{\offinterlineskip
  \tableaustep=\tableauside\advance\tableaustep by-\tableaurule
  \kern\normallineskip\hbox
    {\kern\normallineskip\vbox
      {\gettableau#1 0 }%
     \kern\normallineskip\kern\tableaurule}%
  \kern\normallineskip\kern\tableaurule}}
\def\gettableau#1 {\ifnum#1=0\let\next=\null\else
  \squares{#1}\let\next=\gettableau\fi\next}

\makeatletter

\renewcommand\section{\@startsection{section}{1}{\z@}
                                   {-3.5ex \@plus -1ex \@minus -.2ex}
                                   {2.3ex \@plus .2ex}
                                   {\normalfont\large\bfseries}}
\renewcommand\subsection{\@startsection{subsection}{2}{\z@}
                                   {-3.25ex\@plus -1ex \@minus -.2ex}
                                   {1.5ex \@plus .2ex}
                                   {\normalfont\normalsize\bfseries}}
\renewcommand\subsubsection{\@startsection{subsubsection}{3}{\z@}
                                   {-3.25ex\@plus -1ex \@minus -.2ex}
                                   {1.5ex \@plus .2ex}
                                   {\normalfont\normalsize\bfseries}}
\renewcommand\paragraph{\@startsection{paragraph}{4}{\z@}
                                   {3.25ex \@plus1ex \@minus.2ex}
                                   {-1em}
                                   {\normalfont\normalsize\bfseries}}

\makeatother

\newcommand{\id}{\hbox{1\kern-.27em l}}
\newcommand{\ZZ}{\mathbb{Z}}

\newcommand{\half}{ {\textstyle \frac{1}{2}  } }

\newcommand{\al}{\alpha}

\newcommand{\bet}{\beta}

\newcommand{\ka}{\kappa}

\newcommand{\la}{\lambda}

\newcommand{\tha}{\theta}

\newcommand{\cN}{\mathcal{N}}

\newcommand{\D}{{\rm d}}

\newcommand{\rar}{\rightarrow}
\newcommand{\emp}{\emptyset}
\newcommand{\non}{\nonumber}

\newcommand{\SO}{\mathrm{SO}}

\newcommand{\Sp}{\mathrm{Sp}}
\newcommand{\su}{\mathrm{su}}
\newcommand{\so}{\mathrm{so}}
\newcommand{\spl}{\mathrm{sp}}

\newcommand{\Spin}{\mathrm{Spin}}

\newcommand{\ba}{\begin{array}}
\newcommand{\ea}{\end{array}}
\newcommand{\rso}{rigid semisimple suface operator}

\newtheorem{lemma}{Lemma}[section]
\newtheorem{proposition}{Proposition}[section]

\newtheorem{remark}{Remark}[section]
\newtheorem{Def}{Definition}

\newenvironment{rmk}{\begin{remark} \em}{\end{remark}}

\begin{document}

\title[Rigid Surface Operators and Symbol Invariant of Partitions]
{Rigid Surface Operator and Symbol Invariant of Partitions}

\author[Chuanzhong Li]{Chuanzhong Li}
\address[Chuanzhong Li]{$^1$ College of Mathematics and Systems Science, Shandong University of Science and Technology, Qingdao,266590, P.R.China}
\email{lichuanzhong@sdust.edu.cn}

\author[Bao Shou]{Bao Shou}
\address[Bao Shou]{$^2$ Center  of Mathematical  Sciences,
Zhejiang University,
Hangzhou 310027, China}
\email{bsoul@zju.edu.cn}

\date{}

\begin{abstract}
The symbol is used to describe the Springer correspondence for the classical groups by Lusztig. We refine the explanation  that  the $S$-duality maps of the rigid surface operators are symbol preserving maps. And we find that the maps $X_S$ and $Y_S$ used in the construction of $S$-duality maps are essentially the same. We clear up cause of  the mismatch problem of the total number of  the rigid surface operators between the $B_n$ and $C_n$  theories. And we  construct all the $B_n/C_n$ rigid surface operators which can not have a dual. A classification of  the problematic surface operators is made.
\end{abstract}

\maketitle

\tableofcontents

\section{Introduction}

Surface operators are two-dimensional defects supported on a two-dimensional submanifold of spacetime, which are  natural generalisations of the 't~Hooft operators.  In \cite{GW06},  Gukov and Witten initiated a study of surface operators in $\mathcal{N}=4$ super Yang-Mills theories in the ramified case of the Geometric Langlands Program.

$S$-duality for certain subclass of surface operators is discussed  in \cite{Wit07}\cite{Wy09}. The $S$-duality  \cite{Montonen:1977} assert that
$S : \; (G, \tau) \rightarrow  ( G^{L}, - 1 / n_{\mathfrak{g}} \tau)$ (where $n_{\mathfrak{g}}$ is 2 for $F_4$,  3 for $G_2$, and 1 for other semisimple classical groups \cite{GW06};  $\tau=\theta/2\pi+4\pi i/g^2 $ is usual gauge coupling constant ). This transformation exchanges   gauge group $G$ with the Langlands dual group. For example, the Langlands dual groups of $\Spin(2n{+}1)$ are $\Sp(2n)/\ZZ_2$. And the langlands dual groups of $\SO(2n)$ are  themselves.

In \cite{GW08},  Gukov and Witten extended their earlier analysis \cite{GW06} of surface operators which are based on the invariants of duality. They identified a subclass of surface operators called {\it 'rigid'} surface operators, which are expected to be closed under $S$-duality. There are two types rigid surface operators:  unipotent and semisimple. The rigid semisimple surface operators  are labelled by pairs of  partitions. And unipotent rigid surface operators arise  when one of the  partitions is empty. In  \cite{Wy09}, some proposals for the $S$-duality maps related to rigid surface operators  were made in the $B_n$($\SO(2n{+}1)$) and $C_n$($\Sp(2n)$) theories. These  proposals involved  all unipotent rigid surface operators as well as certain subclasses of rigid semisimple operators.

In  \cite{ShO06},  we  analyse  and extend the $S$-duality maps proposed by Wyllard,  using consistency checks. We  propose the $S$-duality for   a subclasses of  rigid  surface operators.  The symbol invariant is more convenient than other invariants to study the $S$ duality of surface operators but its calculation is boring. In  \cite{Shou-sc}, we propose  equivalent definitions of symbols for different  theories uniformly.  Based on  the new definition,  we simplify the computation of symbol extremely.  We give another construction of the symbol invariant in \cite{SW17}. Fingerprint  is another  invariant of partitions related to the Kazhdan-Lusztig map for the classical groups. We discuss the basic properties of  fingerprint and the constructions in  \cite{SW17}.  We prove the symbol invariant of partitions implies the fingerprint invariant of partitions in \cite{iv}. And we  also make a classification of the symbol preserving maps, which is the basics of  study in this paper.

The $S$ duality maps preserve symbol but not  all symbol preserving maps are $S$ duality maps.    However more thorough understanding   the construction of the $S$ duality of  surface operators might lead to progress.
A problematic mismatch in the total number of rigid surface operators between  the $B_n$ and the $C_n$ theories was pointed out in \cite{GW08} \cite{Wy09}.  The discrepancy is clearly a major problem and hamper the  attempt to  analysis to more general classes of semisimple surface operators.   Fortunately, the construction of symbol \cite{ShO06} and the classification of symbol preserving maps are helpful to address this problem in \cite{iv}.

In this paper,  we attempt to extend the analysis in \cite{GW08}, \cite{Wy09}, and \cite{ShO06}. Since no noncentral rigid conjugacy classes in the $A_n$ theory, we do not discuss surface operators in this case. We also omit the discussion of the exceptional groups, which are more complicated.     We will focus on theories with gauge groups $\SO(2n)$  and the gauge groups $\Sp(2n) $ whose Langlands dual group are $\SO(2n+1)$.

In Section \ref{surf},  we review the construction of rigid surface operators given in \cite{GW08}.  We discuss some mathematical results and definitions as preparation. We focus on the  symbol invariant of  surface operators which are unchanged under the $S$-duality map. In Section \ref{symbol},  we review the  \textit{symbol} invariant proposed in \cite{Wy09},\cite{Shou-sc}. We refine   the computational rules of symbol found in \cite{Shou-sc}. We find the contributions to symbol of a row in the same location of a pairwise rows are the same   in the $B_n$, $C_n$, and $D_n$ theories. As applications,  the $S$-duality maps proposed in the \cite{Wy09} \cite{ShO06} can be illustrated more clearly \cite{Wy09}. We find that the maps $X_S$ and $Y_S$ are   essentially the same map.

The second part of the paper involve the mismatch problem of the total number of  the rigid surface operators between the $B_n$ and $C_n$  theories. We clear up cause of  this problem.
We give the  construction and classification of  all the $B_n/C_n$ rigid surface operators which can not have a dual, revealing some subtle things.

In the appendix, we summarize revelent facts about all rigid surface operators and their associated invariants in the $SO(13)$ and $Sp(12)$ theories as examples.

\section{Surface operators in $\mathcal{N}=4$ Super-Yang-Mills}\label{surf}
In this section, we  introduce the revelent  backgrounds of surface operator.  We closely  follow paper \cite{Wy09} to which we refer the reader for more details.

We consider $\mathcal{N}=4$ super-Yang-Mills theory on $\mathbb{R}^4$ with coordinates $x^0,x^1,x^2,x^3$. The most important bosonic fields: a gauge field as 1-form, $A_\mu$ ($\mu=0,1,2,3$), six real scalars, $\phi_I$ ($I=1,\ldots,6$). All fields take values in the adjoint representation of the gauge group $G$. Surface operators  are introduced  by prescribing a certain singularity structure of  fields near the surface on which the operator is supported. Without loss of generality we can assume the support of the surface operator $D$ to be oriented  along the  $(x^0,x^1)$ directions. Since the fields satisfy the BPS condition,
 the combinations $A=A_2\, \D x^2 +A_3\, \D x^3$ and $\phi = \phi_2 \,\D x^2 + \phi_3 \,\D x^3$ must obey Hitchin's equations \cite{GW08}
\begin{equation}\label{hitch}
 F_A - \phi \wedge \phi = 0,\quad  \D_A\phi =0,\quad \D_A\star A = 0
\end{equation}
A surface operator is defined as   a solution to these equations with a prescribed singularity along the surface $\mathbb{R}^2(x^0,x^1)$.

 For the superconformal surface operator, setting $x_2+ix_3 = re^{i\tha}$,  the most general possible rotation-invariant Ansatz for $A$ and $\phi$ is
\begin{eqnarray}
A &=& a(r) \, \D \tha \,, \nonumber \\
\phi &=& -c(r) \, \D \tha + b(r) \frac{\D r }{r}  \,.
\end{eqnarray}
On substituting this Ansatz into Hitchin's equations (\ref{hitch}) and defining $s = -\ln r$ ,   equations (\ref{hitch}) reduces to Nahm's equations
\begin{eqnarray} \label{nahm}
\frac{\D a}{\D s} &=& [b,c]\,, \nonumber \\
\frac{\D b}{\D s} &=& [c,a] \,,\\
\frac{\D c}{\D s} &=& [a,b] \, \nonumber
\end{eqnarray}
which imply  the communication for the  constants $a$, $b$ and $c$. Surface operators of this type were discussed  in \cite{GW06}.

There is another way to obtain conformally invariant surface operator. Nahm's equations (\ref{nahm}) are solved with
\begin{equation} \label{nahmshou}
a = \frac{t_x}{s + 1/f}\,,\qquad b = \frac{t_z}{s + 1/f}\,,\qquad c = \frac{t_y}{s + 1/f} \,,
\end{equation}
 where $t_x,t_y$ and $t_z$ are elements of the lie algebra $\mathfrak{g}$, spanning a representation  of   $\mathfrak{su}(2)$. These  $t_i$'s are in  the adjoint representation of the gauge group. The surface operator is actually conformal invariant if the function $f$ allowed to fluctuate.

 Alternatively, the surface operators can be characterised    as the conjugacy class of the monodromy
\begin{equation}
U = P \exp(\oint \mathcal{A}) \,,
\end{equation}
where $\mathcal{A} = A + i \phi$. The integration is around a circle near $r=0$.  Following from (\ref{hitch}), one finds that $\mathcal{F} = \D \mathcal{A} + \mathcal{A}\wedge \mathcal{A}=0$, which means that $U$ is independent of   deformations of the integration contour. For the surface operators (\ref{nahmshou}), $U$ becomes
\begin{equation} \label{Uplusshou}
U= P \exp(\frac{2\pi}{s+1/f} \,\,  t_+ ) \,,
\end{equation}
where $t_+\equiv t_x +i t_y$ is nilpotent,  corresponding to  unipotent surface operator.

There are two types of conjugacy classes in a Lie group: unipotent and semisimple. Semisimple classes   can also lead to  surface operators. With a semisimple element $S$ , one can obtain a surface operator with monodromy $V=SU$. For a general  surface operator, it is constructed by requiring  all the fields which are solutions to Nahm's equations satisfy the following constraint  near the surface $D$(\cite{GW08})
\begin{equation} \label{Sun}
S \Psi(r,\tha) S^{-1} = \Psi(r,\tha+2\pi) \,.
\end{equation}

 From all the surface operators constructed from conjugacy classes, a subclass of surface operators called rigid surface operator is  closed on the $S$-duality.  The rigid surface operators are expected to be superconformal and not to depend on any parameters. A unipotent conjugacy classes is called rigid\footnote{The rigid surface operators  here  correspond to   strongly rigid operators in \cite{Wy09}. } if its dimension is strictly smaller than that of any nearby orbit. All rigid orbits have been classified \cite{GW08}\cite{CM93}.  A semisimple conjugacy classes $S$  is called rigid if the centraliser  of such class is larger than that of any nearby class. Summary, surface operators are called  rigid if they  based on monodromies of the form $V=SU$, where   $U$ is unipotent and rigid and $S$ is semisimple and rigid.

\subsection{Preliminary}\label{pre}
From the above discussions,  a classification of unipotent and semisimple conjugacy classes is needed to study surface operators.  Here we  describe the classification of rigid surface operators in the $B_n$($\SO(2n{+}1)$), $C_n$($\Sp(2n)$) and $D_n$($\SO(2n)$) theories in  detail.

The $t_+$ in  Eq.(\ref{Uplusshou}) can be described in  block-diagonal basis  as follows
 \begin{equation}
 \label{ti}
 t_+ = \left( \begin{array}{ccc} t_+^{n_1}  & & \\
                        & \ddots &   \\
                        & & t_+^{n_l}
 \end{array} \right ),
 \end{equation}
where $t_+^{n_k}$ is the `raising' generator of the $n_k$-dimensional irreducible representation of $\su(2)$. For the $B_n$, $C_n$ and $D_n$ theories, there are restrictions on the allowed dimensions of the $\su(2)$ irreps since $t_+$ should belong to the relevant gauge group. From the block-decomposition (\ref{ti}) we see that unipotent (nilpotent) surface operators  are classified by the restricted partitions.

A partition $\lambda$ of the positive integer $n$ is defined by  a decomposition $\sum_{i=1}^l \lambda_i = n$  ($\lambda_1\ge \lambda_2 \ge \cdots \ge \lambda_l$), where the $\lambda_i$ are called parts and  $l$  is the length.
 There is a one-to-one correspondence between partition and Young tableaux. For instance the partition $3^22^31$  corresponds to
\begin{equation}
\tableau{2 5 7}
\end{equation}
The another representation  of partition is $\lambda_m^{n_m}\lambda_{m-1}^{n_{m-1}}\cdots\lambda_1^{n_1}$ with  the length $l=\Sigma_i n_i$ as shown in Fig.(\ref{lm}).
Young diagrams  occur in a number of branches of mathematics and physics. They  are also useful to construct  the eigenstates of Hamiltonian System \cite{Sh11} \cite{{Sh14}} \cite{{Sh15}}.
\begin{figure}[!ht]
  \begin{center}
    \includegraphics[width=2.5in]{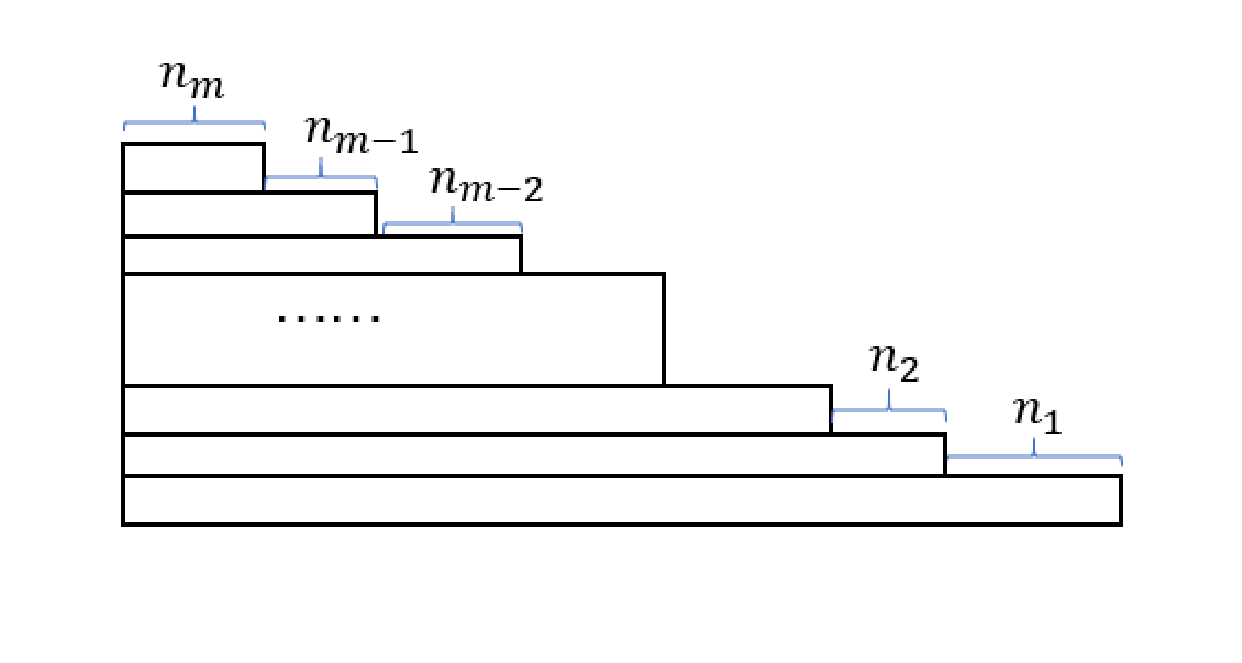}
  \end{center}
  \caption{ Partition  $\lambda_m^{n_m}\lambda_{m-1}^{n_{m-1}}\cdots\lambda_1^{n_1}$ with  the length $l=\Sigma_i n_i$.  }
  \label{lm}
\end{figure}
The addition of two partitions $\lambda$ and $\ka$ is defined by the additions of each part $\lambda_i+\ka_i$.

We have the following classification of nilpotent orbits in terms of partitions\cite{CM93}:
\begin{itemize}
  \item $(B_n)$: partitions of $2n+1$, $\sum \lambda_i=2n+1$, with a constraint that all even integers appear an even number of times;
  \item $(D_n)$: partitions of $2n$, $\sum \lambda_i=2n$, with a constraint that all even integers appear an even number of times;
  \item $(C_n)$: partitions of $2n$, $\sum \lambda_i=2n+1$, with a constraint that all odd integers appear an even number of times;
\end{itemize}
A partition in the $B_n$ or $D_n$($C_n$) theories is called {\it rigid} if it satisfies the following conditions,
\begin{enumerate}
  \item  no gaps (i.e.~$\lambda_i-\lambda_{i+1}\leq1$ for all $i$),
  \item no odd (even) integer appears exactly twice.
\end{enumerate}
Rigid partitions correspond to rigid surface operators. The following facts  are important for studying   rigid partitions, which  are easy to be proved and  omitted here \cite{Wy09}.
\begin{proposition}{\label{Pb}}
The longest row in a rigid  $B_n$ partition always contains an odd number of boxes. And the following two rows of the first row are either both of odd length or both of even length.  This pairwise pattern   then continues. If the Young tableau has an even number of rows the row of shortest length has to be even.
\end{proposition}
\begin{proposition}{\label{Pc}}
The longest two rows in a rigid $C_n$ partition both contain either  an even or an odd number  number of boxes.  This pairwise pattern   then continues. If the Young tableau has an odd number of rows the row of shortest length has contain an even number of boxes.
\end{proposition}
\begin{proposition}{\label{Pd}}
The longest row in a rigid $D_n$ partition always contains an even number of boxes. And the following two rows are either both of even length or both of odd length. This pairwise pattern   then continues. If the Young tableau has an even number of rows the row of the shortest length has to be even.
\end{proposition}

The  rigid semisimple conjugacy classes $S$  in  formula (\ref{Sun}) correspond to diagonal matrices with elements $+1$ and $-1$ along the diagonal in the $B_n$ , $C_n$ and $D_n$ theories\cite{GW08}. The   matrices $S$ break the gauge group to its centraliser at the Lie algebra level as follows
\begin{eqnarray}
\so(2n{+}1) &\rar& \so(2k{+}1)\oplus \so(2n-2k) \,, \nonumber \\
\so(2n) &\rar& \so(2k)\oplus \so(2n-2k) \,, \\
\spl(2n) &\rar& \spl(2k)\oplus \spl(2n-2k) \,, \nonumber
\end{eqnarray}
which  imply  that the rigid semisimple surface operators  correspond to pairs of partitions  $(\lambda';\lambda'')$  in the $B_n$, $C_n$, and $D_n$ \cite{GW08}.  $\lambda'$ is a rigid  $B_k$ partition and  $\lambda''$ is a  rigid $D_{n-k}$ partition in the $B_n$ case.     $\lambda'$ is a rigid $D_k$ partition and  $\lambda''$ is a rigid $D_{n-k}$ partition in the $D_n$ case.  $\lambda'$ is a  rigid $C_k$ partition and  $\lambda''$ is a  rigid $C_{n-k}$ partition in the $C_n$ case.  The rigid unipotent surface operator is a  limiting case of rigid semisimple surface operator with $\lambda''=0$. \footnote{Without confusions, the rigid semisimple surface operators will be called rigid surface operator or surface operator in this study.}

There is a close relationship between the pair of partition $(\lambda';\lambda'')$  and Weyl group.
 For Weyl groups in the $B_n$ , $C_n$,  and $D_n$ theories  both conjugacy classes and irreducible unitary representations are in one-to one correspondence with ordered pairs of partitions $[\al;\bet]$, where  $\al$ is a partition of $n_\al$ and $\bet$ is a partition of $n_{\beta}$, with    $n_\al+n_\bet= n$.
Though both  the conjugacy classes and unitary representations are parameterised by ordered pair of partitions there is no canonical isomorphism between the two sets.

The Kazhdan-Lusztig  map is a  map from the unipotent conjugacy classes of a simple group to the set of conjugacy classes of the Weyl group.
This map can be extended to the case of rigid semisimple conjugacy classes~\cite{Sp92}.
The Springer correspondence is a injective map from the unipotent conjugacy classes of a simple group to the set of unitary representations of the Weyl group. For the classical groups the above  two maps can be described explicitly by the invariants  \textit{fingerprint}  and \textit{symbol} of partitions \cite{CM93}, respectively.

Without explanation,  we only concern   about rigid partition and rigid surface operator  in the following sections.

\subsection{Invariants of surface operators} \label{inv}
Invariants of the surface operators $(\lambda';\lambda'')$  do not change under the $S$-duality map \cite{Wy09} \cite{GW08}.

The dimension $d$  is the most basic invariant of a rigid surface operator. It is calculated as follows \cite{GW08}\cite{CM93}:
\begin{eqnarray}
B_n: & d = 2n^2 + n -\half \sum_{k} (s_k')^2 -  \half \sum_{k} (s_k'')^2
+ \half \sum_{k\;\mathrm{odd}} r_k'+ \half \sum_{k\;\mathrm{odd}} r_k'' \,,\nonumber \\
D_n: & d =2n^2 - n -\half \sum_{k} (s_k')^2 -  \half \sum_{k} (s_k'')^2
+ \half \sum_{k\;\mathrm{odd}} r_k'+ \half \sum_{k\;\mathrm{odd}} r_k''\,, \\
C_n: & d = 2n^2 + n -\half \sum_{k} (s_k')^2 -  \half \sum_{k} (s_k'')^2
- \half \sum_{k\;\mathrm{odd}} r_k'- \half \sum_{k\;\mathrm{odd}} r_k''\,, \nonumber
\end{eqnarray}
where $s'_k$ denotes the number of parts of $\lambda'$'s that are larger than or equal to $k$. And $r_k'$ denotes the number of parts of $\lambda'$ that are equal to $k$. Similarly,  $s_k''$ and $r_k''$ correspond to $\lambda''$.

The invariant  {\it fingerprint } is constructed from $(\lambda';\lambda'')$ via the Kazhdan-Lusztig map. This invariant is a pair of partitions $[\al;\bet]$ associated with  the Weyl group conjugacy class.

There is another invariant {\it symbol} based on the Springer correspondence,  which  can be extended  to rigid semisimple conjugacy classes.  One can construct the symbol of this rigid semisimple surface operator $(\lambda';\lambda'')$ by calculating the symbols for both $\lambda'$ and $\lambda''$,  then add the entries that are `in the same place' of these two partitions. The  result symbol is  denoted as follows
\begin{equation}\label{ddddr}
  \sigma((\lambda^{'};\lambda^{''}))=\sigma(\lambda^{'})+\sigma(\lambda^{''}).
\end{equation}
An example illustrates the addition rule in detail:
\begin{equation} \label{symboladd}
\left(\begin{array}{@{}c@{}c@{}c@{}c@{}c@{}c@{}c@{}c@{}c@{}c@{}c@{}c@{}c@{}} 0&&0&&0&&0&&0&&1&&1 \\ & 1 && 1 && 1 &&1&&1&&2 & \end{array} \right) +
 \left(\begin{array}{@{}c@{}c@{}c@{}c@{}c@{}c@{}c@{}c@{}c@{}c@{}c@{}} 0&&0&&0&&1&&1&&1 \\ & 1 && 1 &&1&&1&&1 & \end{array} \right)=
\left(\begin{array}{@{}c@{}c@{}c@{}c@{}c@{}c@{}c@{}c@{}c@{}c@{}c@{}c@{}c@{}} 0&&0&&0&&0&&1&&2&&2 \\ & 1 && 2 && 2 &&2&&2&&3 & \end{array} \right).
\end{equation}
 It is checked that the symbol of a rigid surface operator contains the same amount of information as the fingerprint \cite{Wy09}. Compared  with the fingerprint invariant, the symbol invariant is much easier to be calculated and  more    convenient   to  find the $S$-duality maps of surface operators.

In \cite{GW08}, it was pointed that two discrete quantum numbers 'center' and 'topology'  are interchanged under $S$ duality. A surface operator can detect topology then its dual should detect the centre and vice versa. However, there are some puzzles using these  discrete quantum numbers to find duality pair   \cite{Wy09}.
There is another problem that   the number of rigid surface operators in the $B_n$ theory is larger than that   in the $C_n$ theory \cite{Wy09}, which  was first observed in the $B_4$/$C_4$ theories \cite{GW08}.
In this paper,  we  ignore the first problem  for the moment. We focus on the symbol invariant to  study the second problem of rigid surface operators between the dual theories. Hopefully, our works will be helpful in making new insight to the surface operator.

\section{Contributions to symbol of rows of partition }\label{symbol}
In this section,  we discuss the contributions to symbols  of  rows of   partitions,  refining  the construction  given in \cite{Shou-sc}.
As applications, we analyse  the $S$ duality maps proposed in  \cite{Wy09},
with a  preparation for the study of the mismatch problem in the $S$-duality map between the number of rigid surface operators in the $B_n$ and $C_n$ theories  in   the next section.

\subsection{Symbol invariant of partitions}

In \cite{Shou-sc}, we  proposed equivalent definitions of symbols for the partitions in the    $C_n$  and $D_n$ theories,  which  are  consistent with that in the $B_n$ theory  as   much as possible.
\begin{Def}\label{dddd}\cite{Shou-sc}
\begin{itemize}
  \item  Symbol  of a partition $\lambda$ in the $B_n$ theory: firstly  add $l-k$  to the $k$th part of the partition $\lambda$.
Then arrange the odd parts and the even parts of the sequence $l-k+\lambda_k$  in  increasing sequences $2f_i+1$ and  $2g_i$, respectively.
Next  calculate the terms
 \begin{equation}\label{aaff}
   \al_i = f_i-i+1\quad\quad\quad  \bet_i = g_i-i+1.
 \end{equation}
 Finally   write the {\it symbol} as
\begin{equation*}
  \left(\ba{@{}c@{}c@{}c@{}c@{}c@{}c@{}c@{}} \al_1 &&\al_2&&\al_3&& \cdots \\ &\bet_1 && \bet_2 && \cdots  & \ea \right).
\end{equation*}

  \item  Symbol  of a partition $\lambda$ in the $C_n$ theory: \begin{description}
                                \item[1]If the length of partition is even,  we compute the symbol as in the $B_n$ case. And then append an extra 0 on the left of the top row of the symbol.
                                \item[2]  If the length of the partition is odd, we append an extra 0 as the last part of the partition.  And then compute the symbol as in the $B_n$ case. Finally,  we delete a 0  in the first entry of the bottom row of the symbol.
                              \end{description}
   \item  Symbol  of a partition $\lambda$ in the $D_n$ theory: we append an extra 0 as the last part of the partition and then compute the symbol as in the $B_n$ case. We  delete  two 0's which occupy the first two entries of the bottom row of the symbol.
\end{itemize}

\begin{rmk}
   Note that the terms $\alpha_*$ in formula (\ref{aaff}) are  related to $f_*$  while the terms $\beta_*$ are  related to $g_*$.
   \end{rmk}

\end{Def}

\begin{figure}[!ht]
  \begin{center}
    \includegraphics[width=4in]{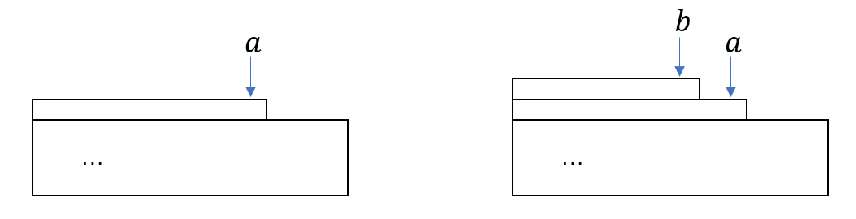}
  \end{center}
  \caption{Addition of  an even row $b$ on the left partition leads to a different partition but in the same theory.  }
  \label{ppf}
\end{figure}
In  \cite{Shou-sc},   we determined  the contribution to symbol for each row of a  partition in the  $B_n(t=-1)$, $C_n(t=0)$ and $D_n(t=1)$ theories, which is  given by Table \ref{tsy}. The table is determined by studying the variation  of  the symbol after adding a row as shown in Fig.(\ref{ppf}). The contribution to symbol of a odd row is calculated  formally (Note that the addition of a odd row leads to a partition in a different theory.  While  the addition of an odd row pairwise rows leads to a partition in the same theory.).
This construction of symbol avoid the boring calculation by  using \textbf{Definition} \ref{dddd},  as shown in the  following example.
\begin{flushleft}
\textbf{Example:}  Symbol of the partition  $\lambda=3^22^21^2$ in the $D_n$ theory,
\end{flushleft}
\begin{equation}
 \tableau{2 4 6}.
\end{equation}
According to Table \ref{tsy},  the symbol is
\begin{equation} \label{1}
\sigma_{{(3^22^21^2)}_{D}}=\left(\begin{array}{@{}c@{}c@{}c@{}c@{}c@{}c@{}}
1&&1&&1 \\ &0&&0& \end{array} \right) +
\left(\begin{array}{@{}c@{}c@{}c@{}c@{}c@{}c@{}} 0&&0&&0 \\ &1&&1& \end{array} \right)+
\left(\begin{array}{@{}c@{}c@{}c@{}c@{}c@{}c@{}} 0&&0&&1 \\ &0&&0& \end{array} \right),
\end{equation}
where the superscript $D$ indicates  that it is  a  partition  in the $D_n$ theory.

\begin{table}
\begin{tabular}{|c|c|c|c|}\hline
Parity of the length of the $i$th row & Parity of $i+t+1$ &  Contribution &  $L$ \\ \hline
odd & even & $\Bigg(\!\!\!\ba{c}0 \;\; 0\cdots \overbrace{ 1\;\; 1\cdots1}^{L} \\
\;\;\;0\cdots 0\;\; 0\cdots 0 \ \ea \Bigg)$   &  $\frac{1}{2}(\sum^{m}_{k=i}n_k+1)$\\ \hline
even & odd & $\Bigg(\!\!\!\ba{c}0 \;\; 0\cdots \overbrace{ 1\;\; 1\cdots1}^{L} \\
\;\;\;0\cdots 0\;\; 0\cdots 0 \ \ea \Bigg)$   &  $\frac{1}{2}(\sum^{m}_{k=i}n_k)$  \\ \hline
even & even &  $\Bigg(\!\!\!\ba{c}0 \;\; 0\cdots 0\;\; 0 \cdots 0 \\
\;\;\;0\cdots \underbrace{1 \;\;1\cdots 1}_{L} \ \ea \Bigg)$  &   $\frac{1}{2}(\sum^{m}_{k=i}n_k)$  \\ \hline
odd & odd    & $\Bigg(\!\!\!\ba{c}0 \;\; 0\cdots 0\;\; 0 \cdots 0 \\
\;\;\;0\cdots \underbrace{1\; \;1\cdots 1}_{L} \ \ea \Bigg)$   &  $\frac{1}{2}(\sum^{m}_{k=i}n_k-1)$  \\ \hline
\end{tabular}
\caption{ Contribution to  symbol of the $i$\,th row   of  the partition   $\lambda_m^{n_m}\lambda_{m-1}^{n_{m-1}}\cdots\lambda_1^{n_1}$. It  depend on the parity of the length of row,  the parameter $L$, and the parity of $i+t+1$  with $t=-1$, $t=0$, and $t=1$ for the partitions in the $B_n$, $C_n$, and $D_n$ theories, respectively. }
\label{tsy}
\end{table}

Using the Table  \ref{tsy},   the calculation of symbol invariant of partition become the combination of  of blocks.  We can  further refine the construction of symbol. Firstly,  we study the contribution to symbol of each row of a pairwise rows of partitions  in different theories. And then we  study  the contribution to symbol of a pairwise rows of  a partition.

For the first step,  we study the contributions to symbol of a row   with  the same location in  a pairwise rows  of      partitions in   different theories.
\begin{figure}[!ht]
  \begin{center}
    \includegraphics[width=4in]{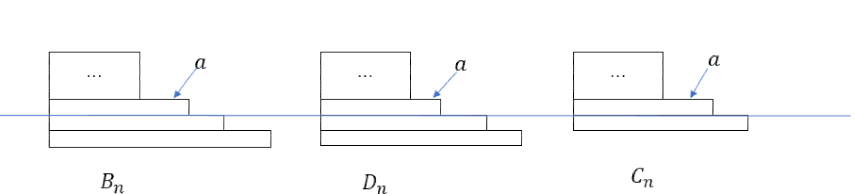}
  \end{center}
  \caption{ Row $a$ is  the top row of   a pairwise rows of partitions in the $B_n$, $D_n$, and $C_n$ theories using Propositions \ref{Pb}, \ref{Pc}, and \ref{Pd}. The row $a$ has the same contribution to symbol in different theories.}
  \label{aa}
\end{figure}
The row $a$ in Fig.(\ref{aa}) is the top row of a pairwise rows  in the $B_n$, $C_n$, and $D_n$ theories according to Propositions \ref{Pb}, \ref{Pc}, and \ref{Pd}.
 \begin{itemize}
   \item  If the  length of the row $a$ is  $2n+1$,   according to  Table \ref{tsy}, its contributions to symbol is
\begin{eqnarray*}
\Bigg(\!\!\!\begin{array}{c} 0 \;\; 0\cdots 0 \;\; {0\cdots0}  \\
\, 0\cdots 0\;\;  \underbrace{1\cdots 1}_{n}\end{array} \!\!\!\Bigg),
 \end{eqnarray*}
 which are the same in different theories.
   \item If the  length of the row $a$ is  $2n$,  according to  Table \ref{tsy}, its contributions to symbol is
\begin{eqnarray*}
\Bigg(\!\!\!\begin{array}{c} 0 \;\; 0\cdots 0 \;\; \overbrace{{1\cdots1}}^{n}  \\
\, 0\cdots 0\;\;  {0\cdots 0}\end{array} \!\!\!\Bigg),
 \end{eqnarray*}
  which are the same in different theories.
 \end{itemize}

 Similarly,  if the row $a$ is at the bottom  of a pairwise rows of a partition, its contribution to symbol is  the same  in different theories.

 Summary,   the same location  of a row in a pairwise rows  partition leads to the same  contribution to symbol  in different theories.

\begin{figure}[!ht]
  \begin{center}
    \includegraphics[width=4in]{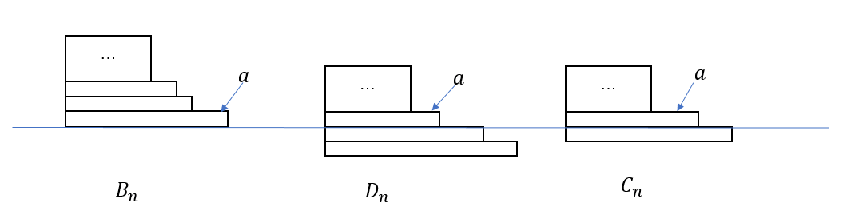}
  \end{center}
  \caption{The first row $a$ of a partition in the $B_n$ theory can be regarded as the top row of an  odd  pairwise rows. With the same location of a pairwise rows,  the row $a$ has the same contribution  to symbol  in  the  $B_n$, $D_n$,  and $C_n$ theories.}
  \label{aaa}
\end{figure}
Secondly, we study the contribution to symbol of the first  row   of      partitions in the   $B_n$  and $D_n$  theories which  do not belong to  a pairwise rows   according to Propositions \ref{Pb} and \ref{Pd}.
Let the row $a$ is the first row of a partition in the $B_n$ theory with length $2n+1$  as shown in Fig.(\ref{aaa}).
According to  Table \ref{tsy}, it  has a contribution to symbol  as follows
\begin{eqnarray}
\Bigg(\!\!\!\begin{array}{c} 0 \;\; 0\cdots 0 \;\; {0\cdots0}  \\
\, 0\cdots 0\;\;  \underbrace{1\cdots 1}_{n}\end{array} \!\!\!\Bigg)
 \end{eqnarray}
which is the same as the contribution to symbol of  the top row   of  an odd  pairwise rows  in the $B_n$, $C_n$, and $D_n$ theories with  length $2n+1$.  Thus we claim that  the first row of a  partition in the $B_n$ theory  can be regarded  as the top row of  an odd  pairwise rows.
\begin{figure}[!ht]
  \begin{center}
    \includegraphics[width=4in]{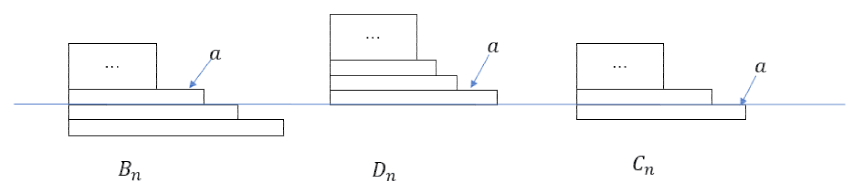}
  \end{center}
  \caption{ The first row  $a$ in the $D_n$ theory can be regarded as the top row of an odd pairwise rows.   With the same location in a pairwise rows,  the row $a$ have the same contributions to symbols  in  the  $B_n$, $D_n$,  and $C_n$ theories.}
  \label{ddd}
\end{figure}

Similarly, we find that the first row of a partition in the $D_n$ theory  can be regarded  as the top row of an even  pairwise rows.  As shown in Fig.(\ref{ddd}),
 the row $a$ with length $2n$  has a contribution to symbol in the $D_n$ theory as follows
\begin{eqnarray}
\Bigg(\!\!\!\begin{array}{c} 0 \;\; 0\cdots 0 \;\; \overbrace{{1\cdots 1}}^{n}  \\
\, 0\cdots 0\;\;  {0\cdots 0}\end{array} \!\!\!\Bigg),
 \end{eqnarray}
which is the same as the contribution to symbol of  the top row with the same length  of  an even  pairwise rows  in the $B_n$,  $C_n$, and $D_n$  theories  according to Table \ref{tsy}.

From the above discussions,  we get the following concise proposition.
\begin{proposition}\label{pair-row}
 With the same location in a pairwise rows of a partition,  one row   has the same  contribution to symbol for partitions in  the  $B_n$, $D_n$,  and $C_n$ theories.
\end{proposition}
\begin{flushleft}
The form of the contribution to symbol of a row of a partition is  shown in Table {\ref{tsyn}}.
\end{flushleft}
\begin{table}
\begin{tabular}{|c|c|c|}\hline
 Location in a pairwise rows & Length   & Contribution    \\ \hline
 top & $2n+1$ & $\Bigg(\!\!\!\ba{c}0 \;\; 0\cdots 0\;\; 0 \cdots 0 \\
\;\;\;0\cdots \underbrace{1 \;\;1\cdots 1}_{n} \ \ea \Bigg)$   \\ \hline
 bottom  & $2n+1$ & $\Bigg(\!\!\!\ba{c}0 \;\; 0\cdots \overbrace{ 1\;\; 1\cdots1}^{n+1} \\
\;\;\;0\cdots 0\;\; 0\cdots 0 \ \ea \Bigg)$    \\ \hline
 bottom  & $2m$  &  $\Bigg(\!\!\!\ba{c}0 \;\; 0\cdots 0\;\; 0 \cdots 0 \\
\;\;\;0\cdots \underbrace{1 \;\;1\cdots 1}_{m} \ \ea \Bigg)$  \\ \hline
 top    & $2m$&  $\Bigg(\!\!\!\ba{c}0 \;\; 0\cdots \overbrace{ 1\;\; 1\cdots1}^{m} \\
\;\;\;0\cdots 0\;\; 0\cdots 0 \ \ea \Bigg)$      \\ \hline
\end{tabular}
\caption{ Contribution to  symbol of a row of a partition  in  the  $B_n$, $D_n$,  and $C_n$ theories.}
\label{tsyn}
\end{table}

\begin{figure}[!ht]
  \begin{center}
    \includegraphics[width=4in]{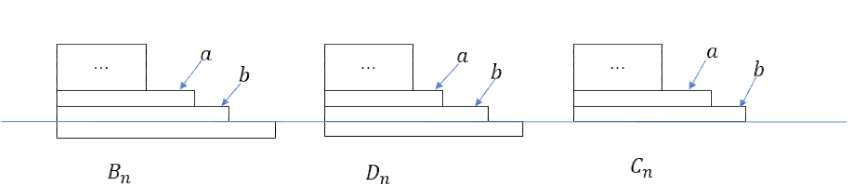}
  \end{center}
  \caption{ Rows $a$ and $b$ form a pairwise rows. The pairwise rows  has the same contribution to symbol in different theories.}
  \label{ab}
\end{figure}
As the applications of Proposition \ref{pair-row}, we study the contributions to symbol of a pairwise rows of a partition. As shown in Fig.(\ref{ab}), the  rows $a$ and $b$ of an odd   pairwise rows    have  the lengths of  $2n+1$ and  $2m+1$, respectively.
According to Table \ref{tsyn},  the pairwise rows has the   contributions to symbol   as follows,
$$\Bigg(\!\!\!\begin{array}{c} 0\;\;0\cdots\cdots 0 \;\;\overbrace{ 1\cdots\cdots 1}^{m+1} \\
\;0 \cdots 0 \;\;\underbrace{1 \cdots 1}_{n}  \ \end{array} \Bigg),
$$
which are the same in the $B_n$, $D_n$, and $C_n$ theories.
Similarly, if the length of $a$ is  $2n$ and the length of $b$  is $2m$,  they have the   contributions to symbol   as follows,
$$\Bigg(\!\!\!\begin{array}{c} 0\;\;0\cdots\cdots 0 \;\; \overbrace{1\cdots\cdots 1}^{n} \\
\;0 \cdots 0 \;\;\underbrace{1 \cdots 1}_{m}  \ \end{array} \Bigg),
$$
which are the same in the $B_n$, $D_n$, and $C_n$ theories.

Summary, we get the following lemma.
\begin{lemma}\label{pair}
A pairwise rows of  partitions   in  the  $B_n$, $D_n$,  and $C_n$ theories has the same  contributions to symbol.
\end{lemma}

\begin{figure}[!ht]
  \begin{center}
    \includegraphics[width=4in]{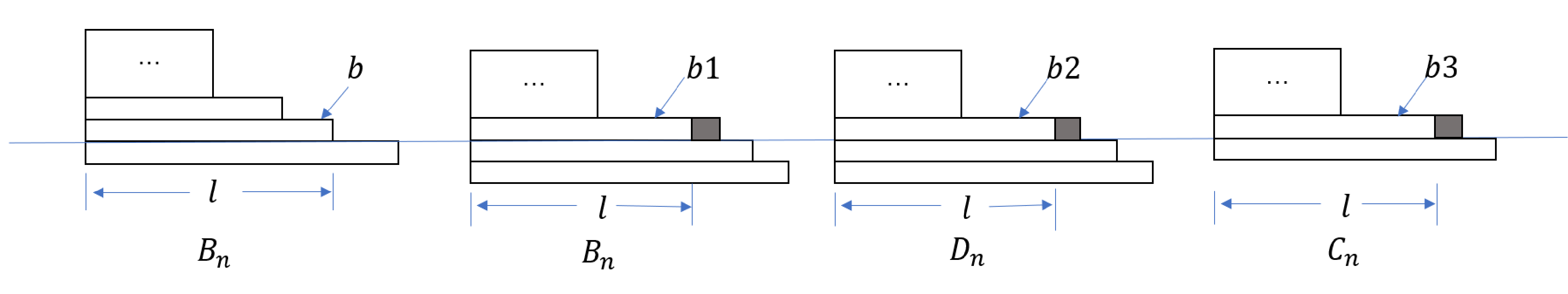}
  \end{center}
  \caption{ Gray boxes are appended at the end of row.  The length of  $b$  is $l$ and the lengths of $b1$, $b2$ and  $b3$ are $l+1$.   The rows $b$, $b1$, $b2$, and $b3$ have  the same contributions to symbol.}
  \label{abab}
\end{figure}
Now  we study the rows of partitions which  have    the same contribution to symbol with different  lengths.  According to Table \ref{tsyn},  the bottom row of an odd pairwise rows has the same contribution to symbol as that of the top row of an even  pairwise rows with one more box. Examples are   shown in Fig.(\ref{abab}). Without an explanation, the gray boxes  denote   the box appended and the black  boxes  denote  the boxes omitted   in the following sections.
The contribution to symbol of the row $b$ with length $2n+1$ in the $B_n$ theory is
\begin{eqnarray}
\Bigg(\!\!\!\begin{array}{c} 0 \;\; 0\cdots 0 \;\;\overbrace{ {1\cdots 1}}^{n+1}  \\
\, 0\cdots 0\;\;  {0\cdots 0}\end{array} \!\!\!\Bigg)
 \end{eqnarray}
which is the same as the contributions of the rows   $b1$,  $b2$, and $b3$ in the $B_n$, $C_n$, and $D_n$  theories, respectively.

\begin{figure}[!ht]
  \begin{center}
    \includegraphics[width=3.5in]{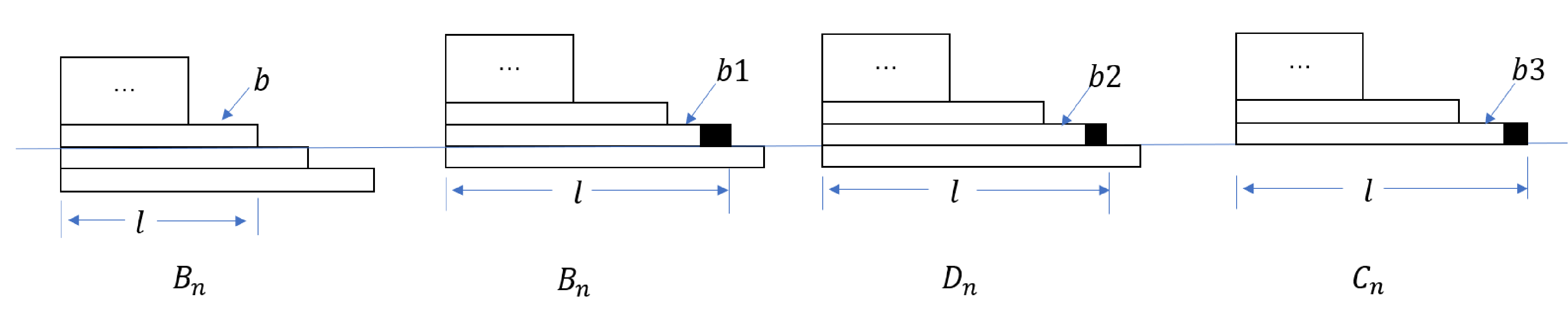}
  \end{center}
  \caption{ Black  boxs are omitted at the end of row. The length of  $b$  is $l$ and the lengths of $b1$, $b2$,  and  $b3$ are $l-1$.   The rows $b$, $b1$, $b2$, and $b3$ have  the same contributions to symbol.}
  \label{ababe}
\end{figure}

According to Table \ref{tsyn}, the top row of  an even  pairwise rows has the same contribution to symbol as that of  the bottom row  of an odd  pairwise rows with one less box. Examples are  shown in Fig.(\ref{ababe}).
The contribution to symbol of the row $b$ with length $2n$ in the $B_n$ theory is
\begin{eqnarray}
\Bigg(\!\!\!\begin{array}{c} 0 \;\;0\cdots 0 \;\;  \overbrace{{1\cdots 1}}^{n}  \\
\, 0\cdots 0\;\;  {0\cdots 0}\end{array} \!\!\!\Bigg)
 \end{eqnarray}
which is the same as that of the rows   $b1$,  $b2$, and $b3$ in the $B_n$, $C_n$, and $D_n$  theories, respectively.

\begin{table}
\begin{tabular}{|c|c|c|}\hline
 Location in a pairwise rows& Length $L$ of row& Contribution  \\ \hline
 top & $2n+1$ & $\Bigg(\!\!\!\ba{c}0 \;\; 0\cdots 0\;\; 0 \cdots 0 \\
\;\;\;0\cdots \underbrace{1 \;\;1\cdots 1}_{n} \ \ea \Bigg)$   \\ \hline
 bottom & $2n$  &  $\Bigg(\!\!\!\ba{c}0 \;\; 0\cdots 0\;\; 0 \cdots 0 \\
\;\;\;0\cdots \underbrace{1 \;\;1\cdots 1}_{n} \ \ea \Bigg)$    \\ \hline
\end{tabular}
\caption{Contribution to symbol of the top row of an odd  pairwise rows with length $2n+1$,  which is the same as the contribution to symbol of the bottom  row of an even  pairwise rows  with length $2n$.   }
\label{tsynu}
\end{table}

\begin{table}
\begin{tabular}{|c|c|c|}\hline
 Location  in a pairwise rows & Length $L$ of row & Contribution \\ \hline
 bottom  & $2n+1$ & $\Bigg(\!\!\!\ba{c}0 \;\; 0\cdots \overbrace{ 1\;\; 1\cdots1}^{n+1} \\
\;\;\;0\cdots 0\;\; 0\cdots 0 \ \ea \Bigg)$    \\ \hline
 top & $2n$ &  $\Bigg(\!\!\!\ba{c}0 \;\; 0\cdots \overbrace{ 1\;\; 1\cdots1}^{n} \\
\;\;\;0\cdots 0\;\; 0\cdots 0 \ \ea \Bigg)$       \\ \hline
\end{tabular}
\caption{ Contribution to symbol of the bottom row of  an odd  pairwise rows with length   $2n+1$,  which is the same as  the contribution to symbol of the top row of  an even  pairwise rows with length $2n$.}
\label{tsyno}
\end{table}
Summary,  we  have the following proposition.
\begin{proposition}\label{row-eo}
 The contribution to symbol of the bottom row of an odd  pairwise rows with length   $L$ is the   same as that of  the top row of an even  pairwise rows  with length $L+1$. And  the contribution to symbol of the top row of  an odd  pairwise rows with length $L$  is the   same as that of  the bottom  row of an  even  pairwise rows with length $L-1$.
\end{proposition}
\begin{flushleft}
This proposition is equivalent to the contents of  Tables {\ref{tsynu}} and  \ref{tsyno}.
 Compared with Table \ref{tsy},  the conclusions of Tables {\ref{tsynu}} and  \ref{tsyno} are not limited to  certain theory.  As shown in Fig.(\ref{ppa}), the rows $a$ and $b$ form a pairwise rows of the first  partition.  The second partition is obtained from the first one by omitting the rows under the row $a$, so it is a partition in the $C_n$ theory.  The third partition is obtained from the first one by omitting the rows under the row $b$, so it is a partition in the $B_n$ or $D_n$ theories, depending on the parity of the length of  row $a$. So  partitions  can be obtained from partitions in different theories.   This picture explain that the row with the same location in a pairwise rows would have the same contribution to symbol in different theories.
\end{flushleft}
\begin{figure}[!ht]
  \begin{center}
    \includegraphics[width=4in]{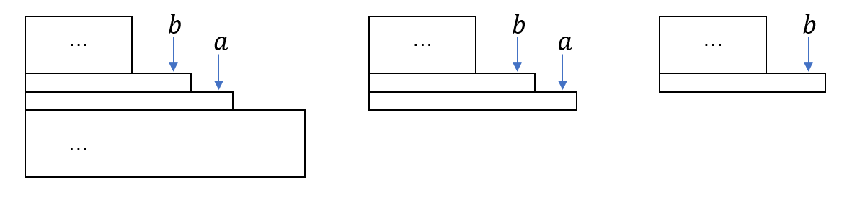}
  \end{center}
  \caption{The rows $a$ and $b$ form a pairwise rows of the first  partition. The other two partitions are obtained from the first one by omitting a part of  rows of the partition. The second partition is in the $C_n$ theory and the third partition would be in the $B_n$ or $D_n$ theories using Propositions \ref{Pb}, \ref{Pc}, and \ref{Pd}. }
  \label{ppa}
\end{figure}

\begin{flushleft}
According to Propositions \ref{pair-row} and \ref{row-eo}, \textit{the contribution to symbol of a  row is  an invariant}. In other words, given the  contribution to symbol of a row,  we can  list out all possible   lengths and  locations  of the row in a pairwise rows. Furthermore, given the symbol invariant, we can list all rigid semisimple surface operators corresponding to the invariant.
\end{flushleft}

\subsection{Maps preserving symbol }\label{spxy}

\begin{figure}[!ht]
  \begin{center}
    \includegraphics[width=5in]{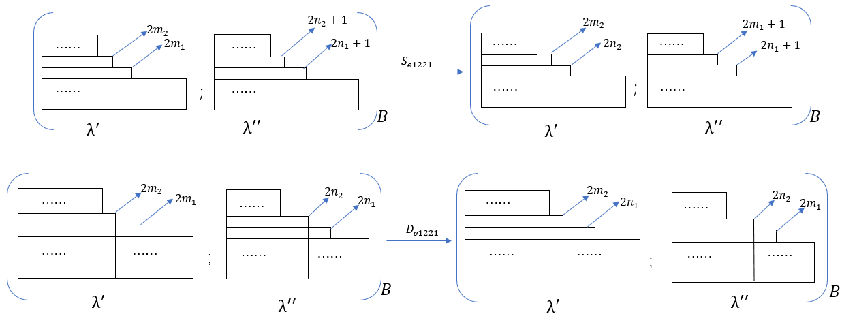}
  \end{center}
  \caption{The subscript $B$ indices a  $B_n$ surface operator with rigid $B_n$ partition    $\lambda^{'}$ and rigid  $D_n$ partition  $\lambda^{''}$.  The map $S_{e1221}$ swap  one row of  $\lambda^{'}$  with one row of  $\lambda^{''}$ in different locations in  pairwise rows.      Without confusion, the partitions of $B_n$ operator  under the maps are denoted as  $\lambda^{'}$ and $\lambda^{''}$ again. And the  map $D_{e1221}$ swap  one row of $\lambda^{'}$  with one row of  $\lambda^{''}$ in the same  location in  pairwise rows. }
  \label{sd}
\end{figure}
There are  two classes of symbol  preserving maps. The first class  of maps  takes  surface operators to surface operators in the same theory.
We have made a classification of the first class of maps in \cite{iv},  with examples  shown in  Fig.({\ref{sd}}). $(\lambda^{'}, \lambda^{''})$ is a rigid semisimple  operator  in the $B_n$ theory.
Under the   map  $S_{e1221}$,   the bottom row of a pairwise rows of $\lambda^{'}$ switches place with the top row of an pairwise row of $\lambda^{''}$ switch places,  which    preserves symbol  according to Proposition \ref{row-eo}.
Under  the map $D_{e1221}$,   the bottom row of a pairwise rows of $\lambda^{'}$  switches place with  the bottom row of an pairwise row of $\lambda^{''}$ switch places,  which preserves symbol according to Proposition \ref{pair-row}.

The second class of maps takes   surface operators to surfaces operator in  different theories,  for examples,  the  $S$ duality maps. Without confusion, the second class of maps will be called the  $S$ duality maps in the following sections.
For the construction  of the $S$ duality maps \cite{Wy09}, the maps $X_S$  and  $Y_S$ play  significant roles. $X_S$ map a partition with only odd rows in the $B_n$ theory to  a partition with only even rows in the $C_n$ theory
\begin{eqnarray} \label{XS}
X_S:&& m^{2n_m+1}\, (m-1)^{2n_{m-1}}\, (m-2)^{2n_{m-2}  } \cdots 2^{2n_2} \, 1^{2n_1}  \non \\   & \mapsto&
m^{2n_m}\, (m-1)^{2n_{m-1} +2 }\, (m-2)^{2n_{m-2} - 2 } \cdots 2^{n_2+2} \, 1^{2n_1-2}\,.
\end{eqnarray}
where $m$ has to be odd in order for the first object to be a partition in the $B_n$ theory.
As shown in  Fig.(\ref{xs}), on the left hand of the map $X_S$, the two rows in braces form   pairwise rows. On the right hand of the map $X_S$, the black boxes are omitted  and the gray boxes  are appended. And the two rows in braces belong to different pairwise rows. The bottom row on the left hand side become the top row on the right hand side while the top row on the left hand side become the bottom row on the right hand side.
\begin{figure}[!ht]
  \begin{center}
    \includegraphics[width=3.5in]{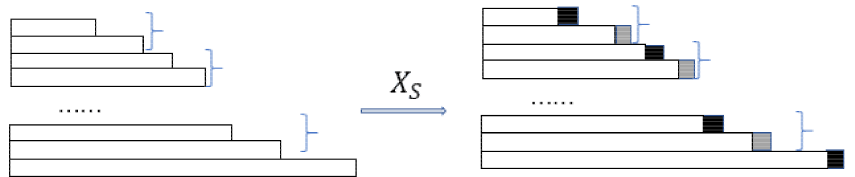}
  \end{center}
  \caption{ On the left hand of the map $X_S$, the two rows in braces form a  pairwise rows. On the right hand of the map $X_S$, the black boxes are omitted  and the gray boxes  are appended. And the two rows in braces belong to different pairwise rows.}
  \label{xs}
\end{figure}

Using Tables \ref{tsynu} and \ref{tsyno}, we can prove  the following lemma directly.
\begin{lemma}\label{sxs}
The map  $X_S$ preserve symbol invariant.
\end{lemma}
\begin{proof}
 On the left hand side of  the map $X_S$, the $2k$th and $(2k+1)$th rows of the partition in the $B_n$ theory form  a pairwise rows excepting the first row. On the other side, the $(2k-1)$th and $2k$th rows  of the partition in the $C_n$ theory  form a pairwise rows. The first row can be regarded  as the top of a pairwise rows.

According to Table \ref{tsynu}, the contribution to symbol of the $2k$th row in the $B_n$ partition is equal to that of the $(2k-1)$th row in the $C_n$ partition.
According to Table \ref{tsyno}, the contribution to symbol of the $(2k+1)$th row in the $B_n$ partition is equal to that of the $2k$th row in the $C_n$ partition.
So the symbols on the two sides of the map $X_S$ are equal.
\end{proof}

\begin{figure}[!ht]
  \begin{center}
    \includegraphics[width=3.5in]{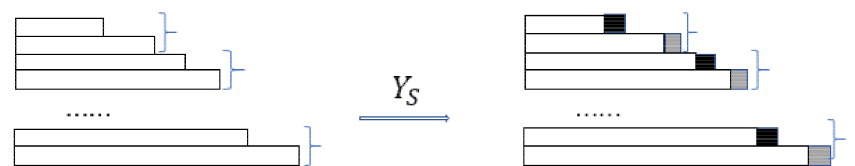}
  \end{center}
  \caption{ On the left hand of the map $Y_S$, the two rows in braces form  a pairwise rows. On the right hand of the map $Y_S$, the black boxes are omitted  and the gray boxes  are appended. And the two rows in braces belong to different pairwise rows.  }
  \label{ys}
\end{figure}
Next, we introduce  the  map $Y_S$ which take  a  rigid partition with only odd rows  in the $C_n$ theory to  a  rigid partition with only even rows  in the $D_n$ theory as shown in Fig.(\ref{ys}).
\begin{eqnarray} \label{YS}
Y_S:&& m^{2n_m+1}\, (m-1)^{2n_{m-1}}\, (m-2)^{2n_{m-2}  } \cdots 2^{2n_2} \, 1^{2n_1}  \non \\   & \mapsto&
m^{2n_m}\, (m-1)^{2n_{m-1} +2 }\, (m-2)^{2n_{m-2} - 2 } \cdots 2^{n_2-2} \, 1^{2n_1+2}\,
\end{eqnarray}
where $m$ has to be even in order for the first element to be a $C_k$ partition. The bottom row on the left hand side become the top row on the right hand side while the top row on the left hand side become the bottom row on the right hand side.
Similarly, we can prove the following lemma.
\begin{lemma}\label{sys}
The map  $Y_S$ preserve symbol invariant.
\end{lemma}

Summary, under the map $X_S$,  we get a partition  $\lambda_{even}$ with only even rows  in the $C_n$ theory from a partition  $\rho_{odd}$ with only odd rows  in the $B_n$ theory,
$$X_S: \rho_{odd}\rightarrow \rho_{even}.$$
Under the map $Y_S$,  we get a partition  $\lambda_{even}$ with only even rows  in the $D_n$ theory from  a partition  $\rho_{odd}$ with only odd rows in the $C_n$ theory,
$$Y_S: \rho_{odd}\rightarrow \rho_{even}.$$
The common  characteristics of the maps  $X_S$ and $Y_S$  are to append  a box at the end of the bottom row of  a pairwise rows and to delete  a box at the end of the top row for a partition with  only odd rows.
Compared Fig.(\ref{xs}) with Fig.(\ref{ys}), the relationship between the map $X_S$ and the map $Y_S$ is
\begin{equation}\label{xyxy}
  X_S(m\rightarrow m-1)=Y_S.
\end{equation}
Thus the map $Y_S$ can be regarded as a  special case of the map $X_S$.\footnote{ The unipotent conjugacy classes (nilpotent orbits) are related to the partitions by Kazhdan-Lusztig map. It would be interesting to study the inspiration of the relationship (\ref{xyxy}) on the nilpotent orbits.} In fact, the Fig.(\ref{ppa}) explain this result.

\subsection{$S$-duality maps for rigid surface operators}\label{sduality}
Combined the   addition rules \ref{ddddr},  the maps $X_S$ and $Y_S$  can be used to construct   the $S$ duality maps of   surface operators.   The $S$ duality maps have the following form
\begin{equation}\label{S}
 S:\,\,\,(\lambda, \rho)_G\rightarrow  (\lambda^{'},\rho^{''})_{G^L}.
\end{equation}
which preserve symbol.
In \cite{Wy09}, Wyllard made  explicit proposals for how the $S$-duality map should act on unipotent surface operators and certain subclasses of semisimple surface operators,  which passe all consistency checks. In \cite{ShO06}, we  made   new proposals for  certain subclasses of semisimple surface operators.

These $S$-duality maps can be explained  naturally as the symbol preserving maps  using Propositions \ref{pair-row} and \ref{row-eo},
drawing  the conclusion directly and avoid complicated derivation.

\begin{flushleft}
\textbf{For rigid unipotent operators $(\lambda, \emptyset)$ in the $B_n$ theory}

The $S$-duality map is
\end{flushleft}
\begin{equation}\label{WB}
 WB:\,\,\,(\lambda, \emptyset)_B\rightarrow (\lambda_{odd}+\lambda_{even},\emptyset)\rightarrow (X_S\lambda_{odd},\lambda_{even})_C.
\end{equation}
Start by splitting  the Young tableau $\lambda$ into  tableau $\lambda_{even}$ constructed from   even rows only  and tableau $\lambda_{odd}$  constructed from the  odd rows only.  Next  the map $X_S$ turns $\lambda_{odd}$ to a partition with only even rows while   $\lambda_{even}$ is left  unchanged. Finally, the duality operator  corresponding to $(\lambda, \emptyset)$ in the  $C_n$ theory is $(X_S\lambda_{odd}, \lambda_{even})$.  According to Proposition \ref{pair-row} and  Lemma \ref{sxs},  the map  $WB$  preserve the symbol.
An example illustrates the procedure.
\begin{flushleft}
  \textbf{Example}: For the $B_{16}$ partition,  $\la = 5\, 4^2 \, 3^3\, 2^4\, 1^3$, applying the map $WB$,  we find
\end{flushleft}
\begin{equation} \label{rank16}
WB:\,\,\, \tableau{1 3 6 10 13} \quad \rightarrow \left(\quad\tableau{4 12}\quad  ;\quad \tableau{6 10}\right)
\end{equation}
which leads to the semisimple $C_{16}$ surface operator $(2^4\,1^8\,,\,2^6\,1^4)$.

%

\begin{flushleft}
\textbf{For rigid unipotent operators $(\lambda, \emptyset)$ in the $C_n$ theory}
Similarly,  the following  $S$-duality map preserve symbol,
\end{flushleft}
\begin{equation}\label{WC}
 WC:\,\,\, (\lambda, \emptyset)_C\rightarrow (\lambda_{odd}+\lambda_{even},\emptyset)\rightarrow (X_S^{-1}\lambda_{even},Y_S\lambda_{odd})_B.
\end{equation}

\begin{flushleft}
\textbf{For semisimple  surface operators $(\rho\,;\rho)$ in the  $C_n$  theory }

The $S$-duality map is
\end{flushleft}

\begin{equation}\label{WCC}
WCC:\,\,\,(\rho\,;\rho)_C\rightarrow(\rho_{\mathrm{even}} +\rho_{\mathrm{odd}} \,;\rho_{\mathrm{odd}} + \rho_{\mathrm{even}}) \rightarrow (\rho_{\mathrm{even}} + X_S^{-1}\rho_{\mathrm{even}}\,;\rho_{\mathrm{odd}} + Y_S\rho_{\mathrm{odd}})_B.
\end{equation}
Firstly,  split two equal tableaux into even-row tableaux $\rho_{\mathrm{even}}$ and odd-row tableaux $\rho_{\mathrm{odd}}$. Then  apply the map $X_S$  to one of the odd-row tableaux and apply the map $Y_S^{-1}$  to the even-row tableau in the other semisimple factor. Next add the altered and unaltered even-row tableaux to form one of the two partitions in a semisimple $B_n$ operator. Finally,  do the same  to the odd-row tableaux and lead to a semisimple operator in the $B_n$ theory.

 $(\rho_{\mathrm{even}} + X_S^{-1}\rho_{\mathrm{even}}\,;\rho_{\mathrm{odd}} + Y_S\rho_{\mathrm{odd}})_B$  is a rigid  surface operator.  An illustration is made through  an example as shown in Fig.(\ref{xy}). A pairwise rows of $\rho_{\mathrm{even}}$ are placed  between the bottom and the top row of a a pairwise rows of   $X_S^{-1}\rho_{\mathrm{even}}$,  not violating the rigid conditions.  A pairwise rows of $\rho_{\mathrm{odd}}$ are placed  between the bottom and the top row of  a pairwise rows of   $Y_S\rho_{\mathrm{odd}}$, not violating the rigid conditions. According to Proposition \ref{pair-row}, the partitions  $\rho_{\mathrm{even}}$ and  $\rho_{\mathrm{odd}}$ have the same contributions to symbol on the two sides of the map $WCC$. According to Proposition \ref{row-eo}, the partitions  $X_S^{-1}\rho_{\mathrm{even}}$ and  $Y_S\rho_{\mathrm{odd}}$ have the same contributions to symbol on the two sides of the map.

\begin{figure}[!ht]
  \begin{center}
    \includegraphics[width=5in]{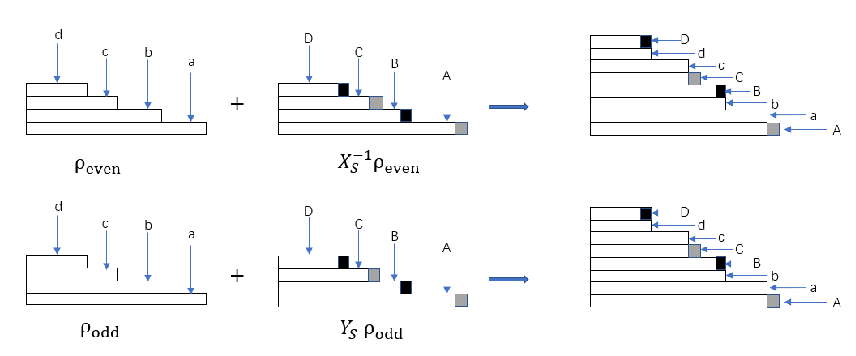}
  \end{center}
  \caption{$\rho_{\mathrm{even}} + X_S^{-1}\rho_{\mathrm{even}}\,$ is a partition in the $B_n$ theory and $\rho_{\mathrm{odd}} + Y_S\rho_{\mathrm{odd}}$ is a partition in the $D_n$ theory.  Thus $(\rho_{\mathrm{even}} + X_S^{-1}\rho_{\mathrm{even}}\,;\rho_{\mathrm{odd}} + Y_S\rho_{\mathrm{odd}})_B$  is a rigid  surface operator in the $B_n$ theory. }
  \label{xy}
\end{figure}

\begin{flushleft}
\textbf{For semisimple  surface operators $(\lambda_{even}\,;\rho_{odd})$ in the  $C_n$  theory }

In \cite{ShO06}, we propose a   $S$-duality map in the sense of symbol invariant as follows,
\end{flushleft}
\begin{equation}\label{cbeo}
 CB_{eo}:\,\,\,(\lambda_{\mathrm{even}}\,;\rho_{\mathrm{odd}})_C\rightarrow(X_S^{-1}\lambda_{\mathrm{even}}\,;Y_S\rho_{\mathrm{odd}})_B,
\end{equation}
which   preserves symbol according to Proposition \ref{row-eo}. One example of this duality is shown in the eighteenth example in the appendix.

\subsection{Discussions}
The $S$ duality maps  preserve symbol invariant and other invariants of partitions. Compared  to other invariants,  the symbol is more easier to be calculated and more convenient to find the $S$-duality maps. Through not all symbol preserving maps are $S$ duality maps,  a more thorough understanding them might lead to progress.  Propositions \ref{pair-row} and \ref{row-eo} make the  the contribution to symbol of rows visualization.   Proposition \ref{pair-row} implies the  symbol preserving operations that   moving a row of a partition to another partition with the same location in  a pairwise rows. One example is that   leaving   $\lambda_{even}$ unchanged in the $S$ duality map $WB$.  Proposition \ref{row-eo} implies the symbol preserving operations such as the maps $X_S$,  $Y_S$ and their inverse maps.  We also find the  important maps $X_S$ and $Y_S$ are essentially the same map.

In fact, the contribution to symbol of a row in a partition  is  also an invariant.  It do not change under the first class of maps and second one.                  Fig.(\ref{ppa}) explain this result.

With these principles in mind, we will discuss the constructions of the rigid  operators in the $B_n$ theory from the $C_n$ theory and vice versa in next section, where the operations  in Propositions \ref{pair-row} and \ref{row-eo} will be used frequently as well as the maps $X_S$,  $Y_S$.

\section{Mismatch of the \rso \, between dual theories}\label{mismatch}

There is a discrepancy of the number of rigid surface operators between the $B_n$ and $C_n$  theories \cite{Wy09}, which  was first observed in the $B_4/C_4$ theories in \cite{GW08}.
Using the generating function for the total number of rigid surface operators(both unipotent and semisimple),  Wyllard found that the difference of number of operators    between  the $B_n$ and  $C_n$ theory is
\begin{equation}\label{num}
  q^9+ 2q^{11}+ 4q^{13}+ 5q^{15}+ 9q^{17}+ 12q^{19}+ 17q^{21}+ 23q^{23}+ \cdots
\end{equation}
where the degree corresponds to the rank $n$ of Lie algebra.

The discrepancy issue is clearly a major problem. Wyllard  gave examples  and made  a preliminary analysis of the problematic surface operators in \cite{Wy09}.  As shown  in the appendix, it seems that
there are two types  of mismatches of  rigid surface operators between the $B_n$  theory and $C_n$  theory.
The first one is  that    certain surface operators  in $B_n/C_n$ theory do not have duals.  And the second one is that the number of surface operators with certain invariants in $B_n$ theory is more than that  in the $C_n$  theory.

In this section,   we analyse the mismatch problem based on constructions of symbol presented  in previous sections.  We find that the discrepancy issue  originates from the rigid conditions of rigid  partitions.

\subsection{Changes of the first row of a partition under $S$ duality}\label{f1}
According to Tables \ref{tsynu} and \ref{tsyno},  the contribution to the symbol of each row of a partition  will not change under the symbol preserving map, which means the contribution to symbol of a row is  an invariant.  So the longest row of the two factors of a rigid surface operator will still  be the longest row on the other side of the $S$-duality map.  According to Propositions \ref{Pb}, \ref{Pc}, and \ref{Pd}, the first two rows of the $C_n$ partitions form   a pairwise rows, while the first row of partitions  in the $B_n$ and $D_n$ theories  not belongs to  a pairwise rows. With these  facts in mind, there are two choices    for the  movements of  the longest row in the second class of   the symbol preserving  maps ($S$-duality maps).
\begin{enumerate}
  \item For the first choice, the longest row  moves from one factor of the \rso \, to the other factor, which will be studied in Sections \ref{di}, \ref{cb}.
  \item For the second one, the longest row stays in the same factor, which will be studied in Sections \ref{m2}.
\end{enumerate}
These two choices correspond to two strategies to construct the $S$-duality maps.

The first class of the symbol maps which is the maps between the \rso \,\,  have been  classified  in \cite{Shou-sc}.  They are   one to one correspondence   on the two side of the $S$-duality map, which will be  illustrated  in Section \ref{otooto}.

\subsection{Generating $B_n$ rigid  semisimple surface operators  from the $C_n$ theory}\label{di}
In this subsection, we propose  algorithms  to generate $B_n$  rigid semisimple surface operators   from that of   the $C_n$ theory.
The two factors of $C_n$  rigid semisimple surface operators are partitions in the $C_n$ theory.  The first two rows of a rigid $C_n$ partition  form a pairwise rows according to Proposition \ref{Pc}.  And thus the parities  of the length of  the first two rows of the  factors of the $C_n$ rigid semisimple surface operator have the same parity or different.

Firstly,  consider the case that  the first two rows of both factors of  the $C_n$ rigid semisimple surface operator  have the same  parities. And there are two cases according to the parity of the length of the first row.
\begin{itemize}
  \item The first two rows of both  factors of a rigid surface operator  are even. The algorithm $EE$ is defined  in Fig.(\ref{EE}). Without lose of generality, we assume the first row of the partition $C2$ is the longest row of the  partitions $C1$ and $C2$. Take  the longest row  from one factor to another one and  append a gray  box at the end of it.  The partition $C1$ become the partition $B1$ and the partition $C2$ become the partition $D2$.
\begin{figure}[!ht]
  \begin{center}
    \includegraphics[width=5in]{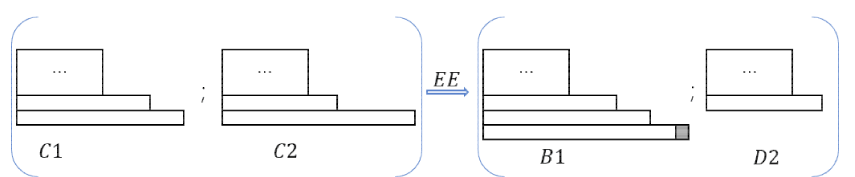}
  \end{center}
  \caption{ Algorithm $EE$ turn a $C_n$  rigid surface operator into a $B_n$  one. The partitions $C1$ and  $C2$ are in the  $C_n$ theory,  with  first two rows   even. And the partitions $B1$ and  $D2$ are in the  $B_n$ and $D_n$ theories, respectively.  }
  \label{EE}
\end{figure}

  \item   The first two rows of both  factors of a rigid surface operator  are odd.  The algorithm $OO$ is defined in Fig.(\ref{OO}). Without lose of generality, we assume the first row of the partition $C2$ is the longest row of the  partitions $C1$ and $C2$. Take  the longest row  from one factor to another one and  append a gray  box at the end of it.  The partition $C1$ become the partition $D2$ and the partition $C2$ become the partition $B1$.
\begin{figure}[!ht]
  \begin{center}
    \includegraphics[width=5in]{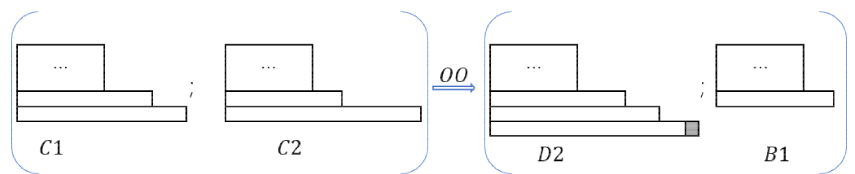}
  \end{center}
  \caption{  Algorithm $OO$ turn a $C_n$  rigid surface operator into a $B_n$  one. The partitions $C1$ and  $C2$ are in the  $C_n$ theory, with the first two rows  odd.. And the partitions $B1$ and  $D2$ are in the  $B_n$ and $D_n$ theories, respectively.  }
  \label{OO}
\end{figure}
\end{itemize}

According to Tables  \ref{tsynu} and \ref{tsyno}, we have the following proposition.
\begin{proposition}\label{eeooee}
The algorithms $EE$ and $OO$ preserve symbol.
\end{proposition}

These algorithms  also preserve the rigid conditions.
\begin{proposition}\label{eeoo}
The algorithms $EE$ and $OO$ preserve rigid conditions of partitions.
\end{proposition}
\begin{proof} We prove the proposition for the algorithm $EE$.
As shown in  Fig.(\ref{EE}), there are  no gaps appearing in the $B_n$ rigid semisimple surface operator $(B1,D2)$.  And the even integers in the partitions $C1, C2$ become the odd integers in the partitions $B1, D2$.  Since no even integer appears exactly twice in the symplectic $(C_n)$ partitions $C1, C2$,  no odd integer  appears exactly twice in the orthogonal $D_n$ partitions $D2$ and no odd integer $(\geq 3)$ appears exactly twice in the orthogonal
$ B_n$ partitions $ B1$. Since the difference of lengths between the longest row appended a gray box and the second row of the partition $B1$ is odd, the part '1' would not appear twice in the  partition $B1$.

Similarly, we can prove the algorithms  $OO$ preserve the rigid conditions of partitions.
\end{proof}

\begin{figure}[!ht]
  \begin{center}
    \includegraphics[width=5in]{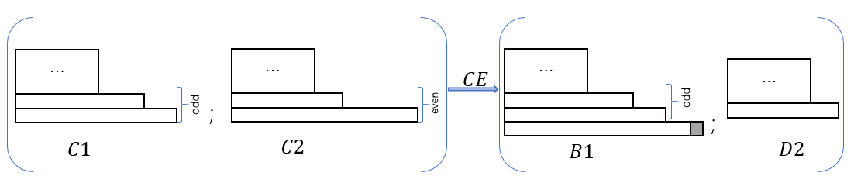}
  \end{center}
  \caption{Algorithms $CE$ turn a $C_n$  rigid surface operator into a $B_n$  one. The first row of  $C2$ is the longest of the two partitions on the left hand side of $CE$.  Add  it to $C1$ and   append a gray box as the last part of the longest row.  }
  \label{CE}
\end{figure}

\begin{figure}[!ht]
  \begin{center}
    \includegraphics[width=5in]{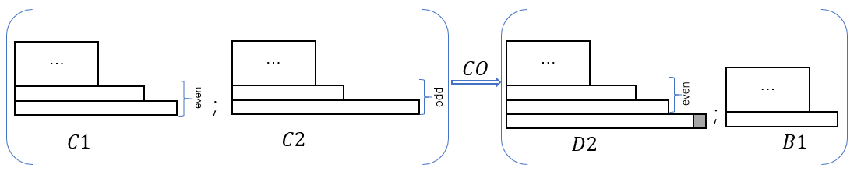}
  \end{center}
  \caption{Algorithms $CO$ turn a $C_n$  rigid surface operator into a $B_n$  one. The first row of  $C2$ is the longest of the two partitions on the left hand side of $CO$.  Add  it  to $C1$ and   append a gray box as the last part of the longest row. }
  \label{CO}
\end{figure}

Secondly,  consider the case that  the first two rows of  factors of  the $C_n$ rigid semisimple surface operator  are of different parities. According to the parity of the length of the longest row,  there are two cases.
\begin{itemize}
  \item The length of the  longest row of two factors  is even. If the first row of  $C2$ is the longest and the length even, we propose an algorithm $CE$ to get a $B_n$ rigid semisimple surface operator from the $C_n$ one as shown in Fig.(\ref{CE}). We add  the longest row to $C1$ and append a gray box, leading to a $B_n$ partition $B1$ and a $D_n$ partition $D2$. The $D_n$ partition $D2$ satisfy the rigid conditions as  Proposition \ref{eeoo}.
  \item The length of the longest row of two factors  is odd.  If the first row of  $C2$ is the longest and the length is odd, we propose an algorithm $CO$  as shown in Fig.(\ref{CO}). We  add  the longest row  to $C1$ and   append a gray box, leading to a $D_n$ partition $D2$ and a $B_n$ partition $B1$.  The $B_n$ partition $B1$ satisfy the rigid  conditions as  Proposition \ref{eeoo}.
\end{itemize}

It is easy  to prove the following proposition  according  to   Tables  \ref{tsynu} and \ref{tsyno}.
\begin{proposition}\label{ceco}
The algorithms $CE$ and $CO$ preserve symbol.
\end{proposition}

However,  under the algorithms $CE$ and $CO$, the partitions $B1$ and $D2$ do not always  preserve the rigid condition.
\begin{flushleft}
 \textbf{ $IC$ type problematic surface operators:} $L(C1)$ and $L(C2)$ denote the lengths of the partitions of $C1$ and $C2$, respectively.
\end{flushleft}
 \begin{itemize}
 \item If $L(C1)=L(C2)-1$, the part '1' appear twice in the  $B_n$  partition $B1$  under the algorithm $CE$, violating   rigid condition (2) in Section \ref{pre}.
 \item If $L(C1)=L(C2)-1$,  the part '1' appear twice in the $D_n$  partition $D2$  under the algorithm $CO$,  violating   rigid condition (2) in Section \ref{pre}.
  \end{itemize}

For these problematic operators, we may try to  add the shorter  row of the first rows of the factors of the $C_n$ rigid semisimple surface operator  from one factor  to the other one. However, these procedures   do not lead to rigid surface operators, violating the rigid condition $\lambda_{i}-\lambda_{i+1}\leq1$   as shown in Figs.(\ref{cos}) and (\ref{ces}).
\begin{figure}[!ht]
  \begin{center}
    \includegraphics[width=5in]{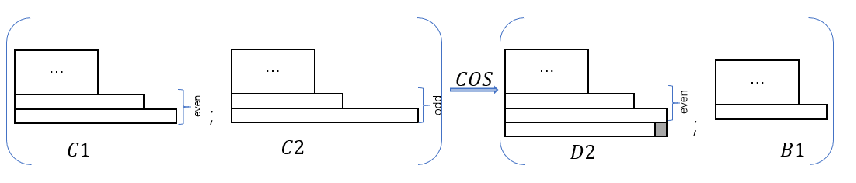}
  \end{center}
  \caption{ $C_n$ partitions $C1$, $C2$.  $D_n$ partition $D2$ and $B_n$ partition $B1$. Algorithm $COS$ add the first row of $C1$   to $C2$ and append a gray box under the condition $L(C1)=L(C2)-1$. }
  \label{cos}
\end{figure}

\begin{figure}[!ht]
  \begin{center}
    \includegraphics[width=5in]{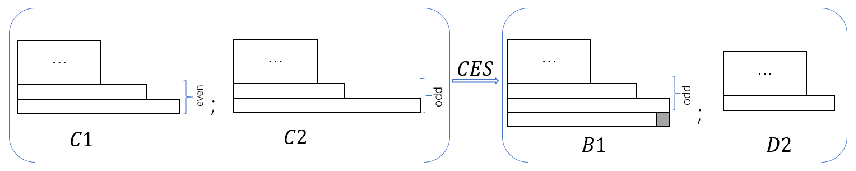}
  \end{center}
  \caption{$C_n$ partitions $C1$, $C2$.  $D_n$ partition $D2$ and $B_n$ partition $B1$. Algorithm $CES$ add the first row of $C1$   to $C2$ and append a gray box under the condition $L(C1)=L(C2)-1$. }
  \label{ces}
\end{figure}
 \begin{itemize}
 \item If $L(C1)=L(C2)-1$,  and then $\lambda_{l-1}-\lambda_{l}=2$ in the $D_n$ partition $D2$  under the algorithm $COS$, violating   rigid condition (1) in Section \ref{pre}.
 \item If $L(C1)=L(C2)-1$,  and then $\lambda_{l-1}-\lambda_{l}=2$ in the $B_n$ partition $B1$ under the algorithm $CES$, violating   rigid condition (1) in Section \ref{pre}.
  \end{itemize}

To  dispel the  obstruction of  the algorithm $CE$,  we may try to  map the $C_n$  operator to another $C_n$  operator with the same symbol  as shown in Fig.(\ref{cbe}) before taking the algorithm $CE$.
\begin{figure}[!ht]
  \begin{center}
    \includegraphics[width=5in]{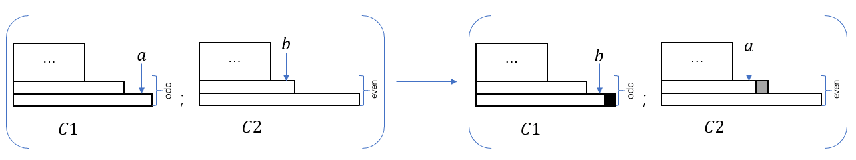}
  \end{center}
  \caption{$C_n$ partitions $C_1$, $C_2$ with   $L(C1)=L(C2)-1$.   }
  \label{cbe}
\end{figure}
We swap the row $a$ of the partition $C1$ with the row $b$ of the partition $C2$ by  deleting the last box of  the row $b$ and appending a box at the end of the row $a$.  However the first two rows of the new factor $C2$ would have the same lengths, violating the rigid condition (1).  We can get the same conclusion for the operator $(C1,C2)$  before taking  the algorithm $CO$.

Summary, the $C_n$ rigid semisimple surface operators $(C1, C2)$ with $|L(C1)-L(C2)|=1$ can not have rigid $B_n$ duals. These problematic surface operators are denoted  as the $IC$ type.

For one class of  the special rigid semisimple surface operator $(\lambda_{even},\lambda_{odd})_C$, there is another strategy to construct the $S$-duality maps.  We will come back this problem in Section \ref{diss}.

\subsection{Generating  $C_n$  rigid semisimple surface operators from the  $B_n$ theory}\label{cb}
The construction of rigid semisimple surface operators  in the $C_n$ theory  from that in the  $B_n$ theory  is  roughly parallel to the discussions in  the last subsection. According to Propositions \ref{Pb} and \ref{Pc}, the first row of the partitions in $B_n$ theory is odd and the first row of the partitions in $D_n$ theory is even.

There are two cases according to the location of the longest row of the factors of the $B_n$ rigid semisimple surface operators.
\begin{itemize}
  \item The longest row of  the rigid semisimple surface operator is  the first row of the  $B_n$ partition $B1$. We suggest the algorithm $BO$ as shown in Fig.({\ref{BC1}}):  delete the last box of the longest row and then  add it to the $D_n$ partition $D2$. Then the first two rows of the $C_n$ partitions  $C2$ are even.    And the partition $C1$ satisfies the rigid condition naturally.
\begin{figure}[!ht]
  \begin{center}
    \includegraphics[width=5in]{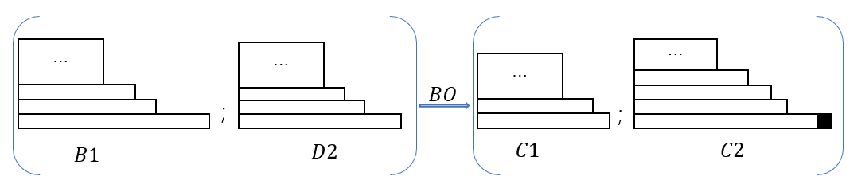}
  \end{center}
  \caption{Partitions $B1$ and $D2$ are in the $B_n$ and $D_n$ theories, respectively. Partitions $C1$ and $C2$ are in the $C_n$ theory. Algorithm $BO$ maps $B_n$ rigid semisimple surface operators to $C_n$ rigid semisimple surface operators. }
  \label{BC1}
\end{figure}
  \item The longest row of the rigid semisimple surface operator is  the first row of the $D_n$ partition $D2$. We suggest the algorithm $BE$ as shown in Fig.({\ref{BC2}}):  delete the last box of the longest row and then  add it to the $B_n$ partition $B1$.   Then the first two rows of the $C_n$ partitions $C1$ are odd.   And the partition $C2$ satisfies  the rigid condition naturally.
  \begin{figure}[!ht]
  \begin{center}
    \includegraphics[width=5in]{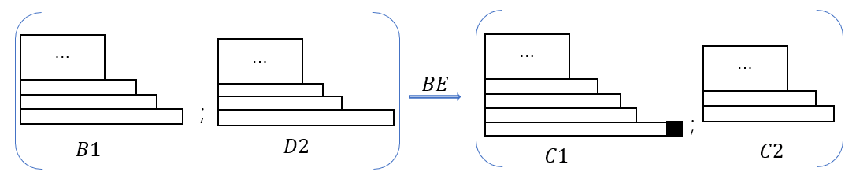}
  \end{center}
  \caption{Partitions $B1$ and $D2$ are in the $B_n$ and $D_n$ theories, respectively. Partitions $C1$ and $C2$ are in the $C_n$ theory. $BE$ map  $B_n$ rigid semisimple surface operators to $C_n$ rigid surface operators.}
  \label{BC2}
\end{figure}
\end{itemize}

However, the partitions $C2$ under the algorithms $BO$ and  $C1$  under the algorithms $BE$ do not always  preserve the rigid conditions.
\begin{flushleft}
 \textbf{ $IB$ type problematic surface operators:} $L(B1)$ and $L(D2)$ denote the lengths of the partitions  $B1$ and $D2$, respectively.
\end{flushleft}

\begin{itemize}
 \item If $L(B1)=L(D2)+1$,  then $\lambda_{l-1}-\lambda_{l}=2$ in the    partition $C2$ under the  algorithm $BO$, violating  the  rigid condition.
 \item If $L(B1)=L(D2)-1$,  then $\lambda_{l-1}-\lambda_{l}=2$ in the   partition $C1$ under the  algorithm $BE$, violating  the  rigid condition.
  \end{itemize}

 To  dispel the  obstruction of  the algorithm $BO$ with $L(B1)=L(D2)+1$,  we may try to  take the $B_n$  operator to another  $B_n$ operator by symbol preserving map  as shown in Fig.(\ref{boe})(a).
We swap the row $a$  with row $b$,  deleting the last box of  the row $b$ and appending a box at the end of the row $a$.  However this operation will not lead to a rigid surface operator since the integer '1' would appear twice in the  $B_n$  partition $B_1$, violating  the  rigid condition.  We  may swap the even row $b$  with even row $c$ as shown in Fig.(\ref{boe})(b).
From the condition $L(B1)=L(D2)+1$, we have  $L(b) \geq L(a)$.
So this operation will not lead to a rigid surface operator $B_1$ in the end.

Similarly, the above operations will not improve the   algorithm $BE$ to get a \rso \,\,\, under the condition $L(B1)=L(D2)-1$.
\begin{figure}[!ht]
  \begin{center}
    \includegraphics[width=5in]{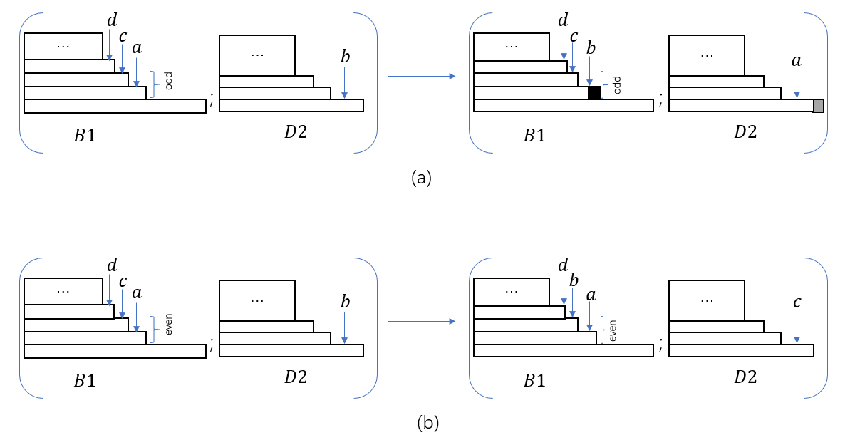}
  \end{center}
  \caption{Maps in Figs.$(a)$ and $(b)$  preserve symbol.  }
  \label{boe}
\end{figure}

Summary, the $B_n$ rigid semisimple surface operators $(B1, D2)$ with $|L(B1)-L(D2)|=1$ can not have rigid $C_n$ duals. These problematic surface operators are denoted  as the $IB$ type.

For the class of the special rigid surface operators $(\lambda_{odd},\lambda_{even})_B$, there is another strategy to construct the $S$-duality maps.  We will  come back to this problem in Section \ref{diss}.

\subsection{One to one correspondence of   maps preserving symbol }\label{otooto}

The second class of symbol  preserving maps is also called $S$-duality  maps, which  take  rigid semisimple surface operator to another  rigid semisimple surface operator in the dual theory. For examples, the algorithms proposed in the last two subsections.
 We find the following relationship between the symbol preserving maps on the two side of these algorithms.
\begin{proposition}\label{oto}
For the algorithms $EE$, $OO$, $CO$, $CE$, $BO$, and $BE$ preserving symbol and  the rigid conditions,  there are  one to one correspondence of the fist class of    symbol  preserving maps  on the two  side of  these algorithms.
\end{proposition}
\begin{figure}[!ht]
  \begin{center}
    \includegraphics[width=5in]{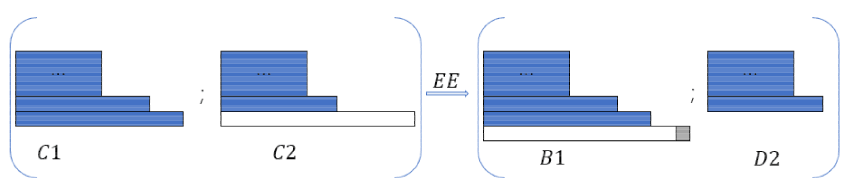}
  \end{center}
  \caption{ Algorithm $EE$ map  $C_n$ rigid operators to $B_n$ rigid operators.  The changes on the blue parts of the rigid surface operators on the left hand side are one to one correspondence to that on the right hand side. }
  \label{oto0}
\end{figure}
\begin{proof}
We prove the proposition for the algorithm  $EE$  as shown in Fig.(\ref{oto0}).  According to the discussions in Section \ref{f1}, for generating rigid semisimple surface operator in the $B_n$ theory from that in the $C_n$ theory, the change of the longest row  is fixed.  The changes are one to  one correspondence between the   blue parts on the two sides  of algorithm  $EE$.

Similarly, we can prove the proposition for the algorithms  $OO$, $CO$, $CE$, $BO$, and $BE$.
\end{proof}
\begin{rmk}
 The algorithms $EE$, $OO$, $CO$, $CE$, $BO$, and $BE$ can be regarded as functors between dual theories,  since they not only map the operators in one theory to that of the dual theory but also map the changes on one side of the algorithms to that of  the other side.
\end{rmk}

We  illustrate this proposition by  two examples as shown in Fig.(\ref{oto1})and Fig.(\ref{oto2}).  The algorithm  $EE$ map the $C_n$ surface operators to the $B_n$ surface operators. The  rows $c11$,  $c12$, $c21$, and $c22$  have the same parities.

For the first example as shown in Fig.(\ref{oto1}), the operation that the rows $c11$ and $c21$ swap places is denoted by down arrow  on the left hand side of the algorithms $EE$, which  leads to  a new rigid semisimple surface operator in the $C_n$ theory.  According to Proposition \ref{pair-row},  this operation preserves symbol and corresponds to the operation swapping $ c11$ with $c21$ denoted by down arrow on  the right hand side of the algorithms $EE$.
\begin{figure}[!ht]
  \begin{center}
    \includegraphics[width=5in]{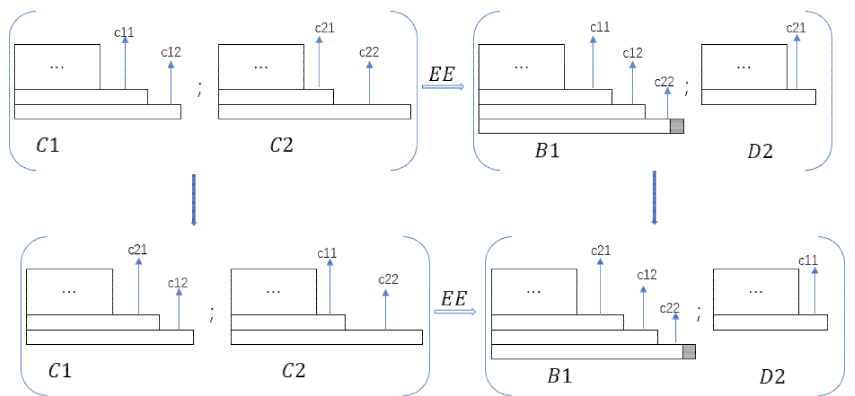}
  \end{center}
  \caption{ Algorithm $EE$ take the  map  preserving symbol of the  $C_n$ rigid surface operator   to that of the  $B_n$ rigid surface operator.}
  \label{oto1}
\end{figure}

For the second example as shown in Fig.(\ref{oto2}),  the row $c21$ of $C2$ is inserted into $C1$. The row $c21$ and rows above it of the partition $C2$   would change parities as well as the   rows above the $c11$ of the partition $C1$.
This operation is denoted   by down arrow on the left hand side of the algorithms $EE$, leading to a new semisimple \rso \, in the same theory.  According to Proposition \ref{row-eo},  this operation  preserve symbol and  corresponds to operation denoted by down arrow on  the right hand side of the algorithms $EE$.
\begin{figure}[!ht]
  \begin{center}
    \includegraphics[width=5in]{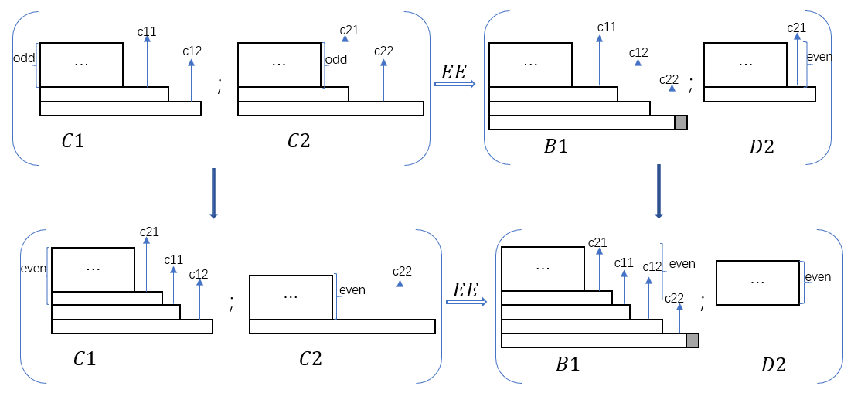}
  \end{center}
  \caption{  Algorithm $EE$ take the  map  preserving symbol of the  $C_n$ rigid surface operator   to that of the  $B_n$ rigid surface operator. }
  \label{oto2}
\end{figure}

As an application,  Proposition \ref{oto}  ensure the  equality  of  the number of rigid surface operators on two sides of  these algorithms ($S$ duality maps).

\subsection{$II$ type problematic surface operators}\label{m2}
Besides  the $IC$ and $IB$  problematic surface operators,   there is another kind of problematic surface operators: the number of     surface operators of one theory  is more than that of the dual  theory with the same symbol invariant. The number of  surface operators in  the $B_n$ theory is one more than that in the $C_n$ theory as shown in  the $18$th and $19$th examples  in the appendix.

\begin{figure}[!ht]
  \begin{center}
    \includegraphics[width=5in]{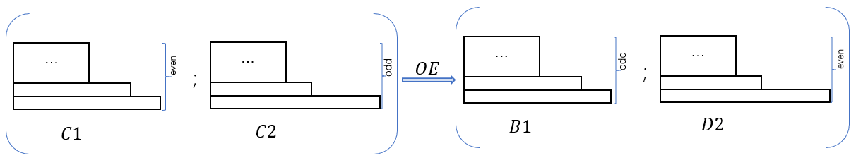}
  \end{center}
  \caption{ Partition  $C1$ with only even rows  and the partition $C2$ with odd rows are in the $C_n$ theory. The partitions  $B1$ with only odd rows and $D2$ with only even rows are in the $B_n$ and $D_n$ theories, respectively. Algorithm $OE$ take the surface operators in the $C_n$ theory to that  in the $B_n$ theory.  }
  \label{noto0}
\end{figure}
This kind of problematic surface operators  appear in the second strategy for the  construction of  the $S$ duality maps.
  $(\lambda_{even}, \rho_{odd})_C$ is  a surface operator in the  $C_n$ theory,  and $\lambda_{even}$ and $\rho_{odd}$ are  partitions with even rows  and odd rows only, respectively.  We take the following algorithm $OE$ to get the $B_n$ rigid semisimple surface operators from that of the $C_n$ theory as shown in Fig.(\ref{noto0}).
$$OE: (\lambda_{even}, \rho_{odd})_C \rightarrow (X_S^{-1}\lambda_{even}, Y_S\rho_{odd})_B \rightarrow (\lambda^{'}_{odd}, \rho^{'}_{even})_B.$$
\begin{figure}[!ht]
  \begin{center}
    \includegraphics[width=5in]{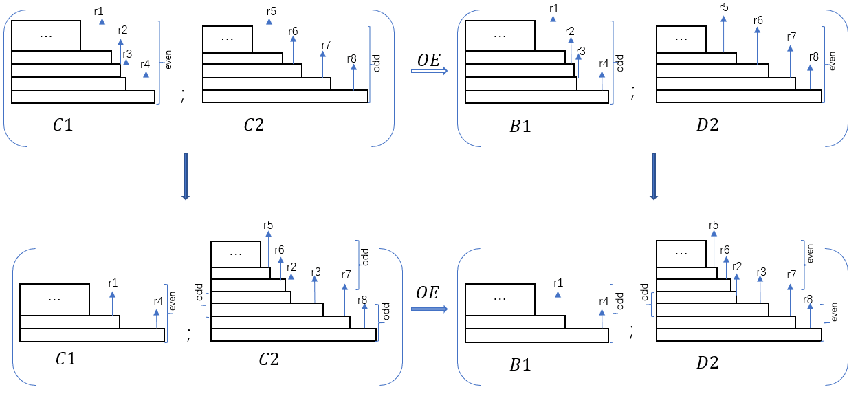}
  \end{center}
  \caption{A pairwise rows  $r2$ and $r3$ of the partition $C1$ are inserted into the partition $C2$.  And A pairwise rows  $r2$ and $r3$ of the partition $B1$ are inserted into the partition $D2$. These two operations are one to one correspondence  under  the algorithm $OE$.   }
  \label{noto1}
\end{figure}

On the other hand,  the algorithm $OE$  as a functor map  the symbol preserving  changes of $C_n$  surface operator $(\lambda_{even}, \rho_{odd})_C $ to that of the $B_n$ one as  shown in Fig.(\ref{noto1}).
However not all the changes on the right hand side of $OE$ could   be realized  on the left hand side. As  shown in Fig.(\ref{noto2}), the even row $r1$ is the top row of  a pairwise rows of the partition $C1$, and the odd row $r2$ is the bottom row of  a pairwise rows  of the partition $C2$. The length of $r1$ is shorter than that of the row $r2$. Under the algorithm $OE$, to preserve the symbol,  the row $r1$ becomes odd and becomes the bottom row of  a pairwise rows of  $B1$, and the row $r2$ becomes even and becomes the top row of  a pairwise rows of  $D2$.
Now we  take the   $B_n$ \rso \,  $(\lambda^{'}_{even}, \rho^{'}_{odd})$ to another $B_n$ \rso \, under the down arrow on the right hand side of $OE$.
We  put the  $r1$ and the parts  above it  above  $r2$ of $D2$. 
This  change of the  $B_n$ \rso \,  $(X_S^{-1}\lambda_{even}, Y_S\rho_{odd})$   can not be realized in the  $C_n$ \rso \,  $(\lambda_{even}, \rho_{odd})$.  Assume $r1$ is putted above $r2$ on the left hand side of $OE$, corresponding to the down arrow on the left  hand side of $OE$ as shown in Fig.(\ref{noto2}),  then they form a pairwise rows with different parities,  which is a contradiction.
\begin{figure}[!ht]
  \begin{center}
    \includegraphics[width=5in]{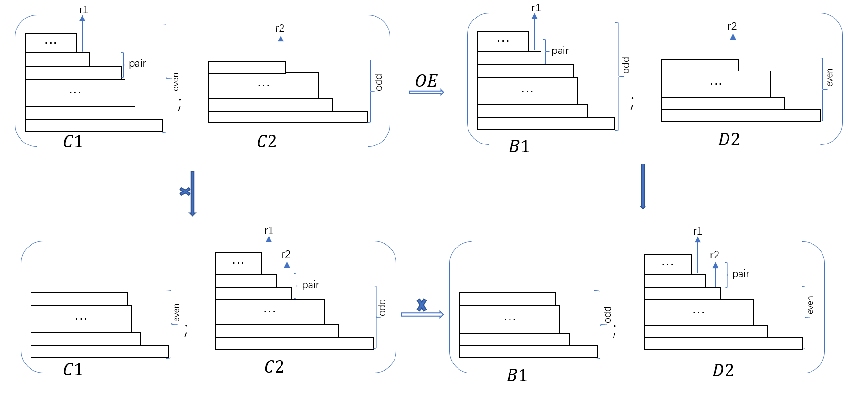}
  \end{center}
  \caption{Algorithm $OE$ take  the $C_n$ \rso \,   \,\,$(\lambda_{even}, \rho_{odd})$ to the $B_n$ \rso \,   \,\,$(X_S^{-1}\lambda_{even}, Y_S\rho_{odd})$. The row $r1$ and the rows above $r1$ of $B1$ are placed  upon the row  $r2$ of $D2$ under the algorithm $OE$.  The operation corresponding to the down arrow on the right hand of $OE$ fail to  be realized on the left hand side  of   $OE$.  }
  \label{noto2}
\end{figure}

The  algorithms  in Figs.(\ref{noto1}) and  (\ref{noto2}) are particularly revealing.  For  the $C_n$ operators $(\lambda_{even}, \rho_{odd})$,   the algorithm $OE$ will work  when the algorithm in Sections \ref{di}   fail to preserve rigid conditions.
Since not all the symbol preserving maps  of surface operators on the right  side of $OE$ can be  realized on the left side,  the number of  rigid $B_n$  surface operators is more than  that of   the $C_n$ surface operators with the symbol invariant of the \rso \, in Fig(\ref{noto2}). We denote them as the $IIC$ type problematic surface operators.

Similarly, we can propose an algorithm $OE$ to get  $C_n$ rigid semisimple surface operators from that of the $B_n$ theory as follows
$$EO:(\lambda^{'}_{odd}, \rho^{'}_{even})_B \rightarrow (X_S\lambda^{'}_{odd}, Y^{-1}_S\rho^{'}_{even})_C\rightarrow (\lambda_{even}, \rho_{odd})_C .$$
And we come to the  conclusion    that  the number of  rigid $C_n$  surface operators is more than  that of   the $B_n$ surface operators for certain symbol invariant.
We denote them as the $IIB$ type problematic surface operators.

\subsection{Generating $D_n$  rigid semisimple surface operator from the $D_n$ theory}
 Since the langlands dual groups of $\SO(2n)$ are  themselves, the $D_n$ theory is self duality. The first class of symbol preserving maps is the same with the second class of symbol preserving map. Thus the  $S$-duality pairs can be  realized by the first  class of symbol preserving maps,  which do not lead to semisimple  surface operators        violating rigid conditions.
 For the certain symbol invariant  with only one \rso \,, we suggest  the following $S$ duality map
$$\mathbf{1}:(\lambda, \rho)_D \rightarrow (\lambda, \rho)_D,$$
which map a rigid surface operator to itself.

\subsection{Classification of problematic surface operators and discussions}\label{diss}

\begin{table}
\begin{tabular}{|c|c|c|c|c|c|}\hline
Type  & Theory  & Surface operator  & Conditions & Algorithms \\ \hline
IC& $C_n$  & $(\lambda, \rho)_C$  &$|L(\lambda)-L( \rho)|=1$,\,$(\lambda, \rho)_C\neq (\lambda_{odd}, \rho_{even})_C$ & $CE$, $CO$ \\ \hline
IB & $B_n$   &  $(\lambda, \rho)_B$ & $|L(\lambda)-L( \rho)|=1$,\,$(\lambda, \rho)_B \neq (\lambda_{odd}, \rho_{even})_B$& $BE$, $BO$   \\ \hline
IIC & $C_n$   &  $(\lambda, \rho)_C$  & $(\lambda, \rho)_C=(\lambda_{odd}, \rho_{even})_C$ & $OE$ \\ \hline
IIB & $B_n$  &   $(\lambda, \rho)_B$  &$(\lambda, \rho)_B = (\lambda_{odd}, \rho_{even})_B$& $EO$  \\ \hline
\end{tabular}
\caption{ Classification of  problematic surface operators. The subscripts $odd $ and $even$ mean partitions with only odd and even rows, respectively.  }
\label{ms}
\end{table}
We  find  new type of  of problematic operators  excepting  all the ones  given  in \cite{Wy09}
Even more, the algorithms proposed give all the problematic rigid surface operators. The classification of the problematic surface operators  in the previous sections is given  by Table \ref{ms}. The discussions in Section \ref{f1} and Proposition \ref{oto} ensure it is a completeness classification.

\begin{figure}[!ht]
  \begin{center}
    \includegraphics[width=2in]{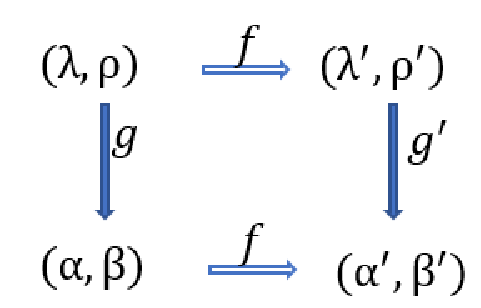}
  \end{center}
  \caption{ $g$ and $g'$ belong to  the first class of symbol preserving maps which map surface operators from one theory to itself.  $f$   is the second class of map ($S$-duality map) which  take surface operator from one theory to the dual theory as the algorithms $EE, OO$.  The map $f$ also map the first class of map $g$ to $g'$.  }
  \label{prsu}
\end{figure}
We find two  types of problematic surface operators in this study which are denoted as   $I$ and $II$.  For the non-problematic surface operators, we have the  commutation relation as shown in Fig.(\ref{prsu}).
The $I$ type  surface operators exists only in one theory which mean the map $f$ can not work in Fig.(\ref{prsu}).  For example, the algorithm $BO$ can not map the operators in $B_n$ theory to that in $C_n$ theory with restriction  conditions given in Table (\ref{ms}).
And the $II$ type one are surface operators which  do not have the same number of operators in dual theories with certain symbol. It means that the map $g$ can be only realized in one theory as shown in Fig.(\ref{noto2}). The origin of both types of problematic surface operators are the rigid conditions.


We can learn much from Table \ref{ms}.
When the  algorithms $CE$ and $CO$ work, they would realize all the $S$ duality pairs with certain symbol.
When the algorithms $CE$ and $CO$  fail to realize the $S$ duality pairs, the algorithm $OE$ is the only choice to work,  which   is an evidence of the $S$ duality map $CB_{eo}$ (\ref{cbeo}).


From  formula (\ref{num}),  one gets  further insight into the mismatch problem.  The coefficient is positive, which imply that the number of rigid surface operators  in the $B_n$ theory is larger  than that in  the $C_n$ theory. A naive gauss would be that there are more  $B_n$ surface operators than $C_n$ surface operators  with the given  symbol. In fact, in \cite{Wy09},  it only point out  that the number of the $B_n$ surface operators is more than the number of   the $C_n$ surface operators. They do not find that there are   rigid surface operators in the $C_n$ theory  which do not have candidate duals in the $B_n$ theory.
However, according to Table \ref{ms}, the $IC$  type $C_n$  problematic     surface operators  can not have duals in the $B_n$ theory under the algorithms $BE$ and $BO$.  They also did not find the  $IIC$ type problematic surface operators  in the $C_n$ theory.

The number of  the rigid surface operators which do not have candidate duals in the $C_n$ theory do increase with the  rank $n$ from the discussion in Section \ref{di}. Fortunately,   the excess number of states divided by the total number  appears to approach zero as $n\rightarrow\infty$. So one  hopes that only a minor modification is needed to make the numbers match, which is consistent  with the fact that most rigid surface operators do seem to have candidate duals.

The physical reason for the discrepancy is still unknown. Throughout this paper we will only consider strongly rigid operators which we refer to as rigid surface operator.  From the discussions,  we  should also take account of the larger class including the weakly rigid surface operators discussed in \cite{Wy09} or the quantum  effect to resolve  the mismatch in the total number of rigid surface operators. Clearly more work is required.

Furthermore,  the construction of symbol invariant can be used  to study the $S$ duality of the rigid  surface operators in   other  Langlands dual groups  such as exception Lie algebra, and the ??research  problems related to the Springer correspondence.

\section*{Acknowledgments}
We would like to thank  Zhisheng Liu and Qi Li for  many helpful discussions.
Chuanzhong Li is supported by the National Natural Science Foundation
of China under Grant No.12071237.


\appendix

\section{Rigid semisimple  surface operators in $\SO(13)$ and $\Sp(12)$ }\label{tb}
 The first column is the type of  the duality maps listed  in \cite{ShO06} .  The second and third columns   list   pairs of partitions corresponding to the surface operators in the $B_n$ and $C_n$ theories. The other  columns are  the  dimension,   symbol invariant,   and  fingerprint invariant of the surface operator, respectively.  Even the   mismatch in the total number of rigid surface operators in the $B_n$ and $C_n$ theories can be explained. The $18$th and $19$th pairs of rigid semisimple surface operators belong to the $II$ type mismatch. The $20$th, $23$th, and $24$th pairs of rigid semisimple surface operators belong to the $I$ type mismatch.

\begin{equation} {
\begin{array}{l@{\hspace{10pt}}l@{\hspace{10pt}}l@{\hspace{10pt}}l@{\hspace{10pt}}l@{\hspace{10pt}}c@{\hspace{10pt}}l}
\underline{Num}
&
\underline{Sp(12)}
&
\underline{SO(13)}
&
\underline{Dim}
&
\underline{Symbol}
&
\underline{Fingerprint}
\\[-0.5pt]
1
&
(1^{12}\,;\emp)
&
(1^{13};\emp)
 &
 0
 &
\left(\begin{array}{@{}c@{}c@{}c@{}c@{}c@{}c@{}c@{}c@{}c@{}c@{}c@{}c@{}c@{}c@{}}
 0&&0&&0&&0&&0&&0&&0 \\
 &1&&1&&1&&1&&1&&1&
\end{array}\right)
 &
[1^6;\emp]
 \\[-0.5pt]
 2
&
 (2\,1^{10}\,;\emp)
 &
(1;1^{12})
 &
 12
 &
\left(\begin{array}{@{}c@{}c@{}c@{}@{}c@{}c@{}c@{}c@{}c@{}c@{}c@{}c@{}}
 1&&1&&1&&1&&1&&1 \\
 &0&&0&&0&&0&&0&
\end{array}\right)
 &
[1^5;1]
 \\[-0.5pt]
3
&
  (1^{10}\,;1^2)
 &
(2^2\,1^9;\emp)
 &
 20
 &
\left(\begin{array}{@{}c@{}c@{}c@{}c@{}c@{}c@{}c@{}c@{}c@{}c@{}c@{}}
 0&&0&&0&&0&&0&&0 \\
 &1&&1&&1&&1&&2&
\end{array}\right)
 &
[2\,1^4;\emp]
 \\[-0.5pt]
 4
&
 (2^3\,1^6\,;\emp)
 &
(1;2^2\,1^8)
 &
 30
 &
\left(\begin{array}{@{}c@{}c@{}c@{}@{}c@{}c@{}c@{}c@{}c@{}c@{}}
 1&&1&&1&&1&&1 \\
 &0&&0&&0&&1&
\end{array}\right)
 &
[1^3;1^3]
\\[-0.5pt]
5
&
(2\,1^8\,;1^2)
 &
(1^3;1^{10})
 &
 30
 &
\left(\begin{array}{@{}c@{}c@{}c@{}@{}c@{}c@{}c@{}c@{}c@{}c@{}c@{}c@{}}
 1&&1&&1&&1&&1 \\
 &0&&0&&0&&1&
\end{array}\right)
 &
[1^3;1^3]
 \\[-0.5pt]
 6
&
 (1^{8}\,;1^4)
 &
(2^4\,1^5;\emp)
 &
 32
 &
\left(\begin{array}{@{}c@{}c@{}c@{}c@{}c@{}c@{}c@{}c@{}c@{}}
 0&&0&&0&&0&&0 \\
 &1&&1&&2&&2&
\end{array} \right)
 &
[2^2\,1^2;\emp]
 \\[-0.5pt]
7
&
 (2^4\,1^4\,;\emp)
 &
(3\,2^2\,1^6;\emp)
 &
 36
 &
\left(\begin{array}{@{}c@{}c@{}c@{}c@{}c@{}c@{}c@{}c@{}c@{}}
 0&&0&&0&&1&&1 \\
 &1&&1&&1&&1&
\end{array}\right)
 &
[1^2;1^4]
 \\[-0.5pt]
 8
&
 (1^8\,;2\,1^2)
 &
(1^9,1^4)
 &
 36
 &
\left(\begin{array}{@{}c@{}c@{}c@{}c@{}c@{}c@{}c@{}c@{}c@{}}
 0&&0&&0&&1&&1 \\
 &1&&1&&1&&1&
\end{array}\right)
 &
[1^2;1^4]
 \\[-0.5pt]
9
&
 (1^6\,;1^6)
 &
(2^6\,1;\emp)
 &
 36
 &
\left( \begin{array}{@{}c@{}c@{}c@{}c@{}c@{}c@{}c@{}}
 0&&0&&0&&0 \\
 &2&&2&&2&
\end{array} \right)
 & [2^3;\emp]
 \\[-0.5pt]
 10
&
 (2^5\,1^2\,;\emp)
 &
(1;2^4\,1^4)
 &
 40
 &
\left(\begin{array}{@{}c@{}c@{}c@{}@{}c@{}c@{}c@{}c@{}}
 1&&1&&1&&1 \\
 &0&&1&&1&
\end{array}\right)
 &
[1;1^5]
 \\[-0.5pt]
 11
&
  (2\,1^6\,;1^4)
 &
(1^5;1^8)
 &
 40
 &
\left(\begin{array}{@{}c@{}c@{}c@{}@{}c@{}c@{}c@{}c@{}c@{}c@{}}
 1&&1&&1&&1 \\
 &0&&1&&1&
\end{array}\right)
 &
[1;1^5]
 \\[-0.5pt]
 12
&
  (1^6\,;2\,1^4)
 &
(1^7;1^6)
 &
 42
 &
\left(\begin{array}{@{}c@{}c@{}c@{}c@{}c@{}c@{}c@{}}
 0&&1&&1&&1 \\
 &1&&1&&1&
\end{array}\right)
 &
[\emp;1^6]
 \\[-0.5pt]
 13
&
(3^2\,2\,1^4\,;\emp)
 &
(1^3;2^2\,1^6)
 &
 44
 &
\left(\begin{array}{@{}c@{}c@{}c@{}@{}c@{}c@{}c@{}c@{}c@{}c@{}}
 1&&1&&1&&1 \\
 &0&&0&&2&
\end{array}\right)
 &
[3\,1^2;1]
 \\[-0.5pt]
 14
&
  (2^3\,1^4\,;1^2)
 &
(2^2\,1;1^8)
 &
 44
 &
\left(\begin{array}{@{}c@{}c@{}c@{}@{}c@{}c@{}c@{}c@{}c@{}c@{}}
 1&&1&&1&&1 \\
 &0&&0&&2&
\end{array}\right)
 &
[3\,1^2;1]
 \\[-0.5pt]
 15
&
(2\,1^6\,;2\,1^2)
 &
(1;3\,2^2\,1^5)
 &
 44
 &
\left(\begin{array}{@{}c@{}c@{}c@{}@{}c@{}c@{}c@{}c@{}}
 1&&1&&2&&2 \\
 &0&&0&&0&
\end{array}\right)
 &
[2\,1^2;2]
 \\[-0.5pt]
 16
&
  (2^4\,1^2\,;1^2)
 &
(2^21^5;1^4)
 &
 48
 &
\left(\begin{array}{@{}c@{}c@{}c@{}c@{}c@{}c@{}c@{}}
 0&&0&&1&&1 \\
 &1&&1&&2&
\end{array}\right)
 &
[3\,1;1^2]
 \\[-0.5pt]
 17
&
  (2\,1^4\,;2\,1^4)
 &
(1;3\,2^4\,1)
 &
 48
 &
\left(\begin{array}{@{}c@{}c@{}c@{}@{}c@{}c@{}}
 2&&2&&2 \\
 &0&&0&
\end{array}\right)
 &
[2^2;2]
 \\[-0.5pt]
 18
&
  (2^3\,1^2\,;1^4)
 &
(1^5;2^2\,1^4)
 &
 50
 &
\left(\begin{array}{@{}c@{}c@{}c@{}c@{}c@{}}
  1&&1&&1 \\
 &1&&2&
\end{array}\right)
 &
[3;1^3]
 \\[-0.5pt]
 19
&
-
 &
 (2^21^3;1^6)
&
 50
 &
\left(\begin{array}{@{}c@{}c@{}c@{}@{}c@{}c@{}c@{}c@{}}
 1&&1&&1 \\
 &1&&2&
\end{array}\right)
 &
[3;1^3]
 \\[-0.5pt]
 20
&
-
 &
(2^4\,1;1^4)
 &
 52
 &
\left(\begin{array}{@{}c@{}c@{}c@{}c@{}c@{}}
 0&&1&&1 \\
 &2&&2&
\end{array}\right)
 &
[3^2;\emp]
\\[-0.5pt]
 21
&
 (2^3\,1^2\,;2\,1^2)
 &
(1^3;3\,2^2\,1^3)
 &
 54
 &
\left(\begin{array}{@{}c@{}c@{}c@{}@{}c@{}c@{}c@{}c@{}}
  1&&2&&2 \\
  &0&&1&
\end{array}\right)
 &
[3\,1;2]
\\[-0.5pt]
22^{*}
&
 (3^2\,2\,1^2\,;1^2)
&
(2^2\,1;2^2 \,1^4)
 &
 54
 &
\left(\begin{array}{@{}c@{}c@{}c@{}@{}c@{}c@{}c@{}c@{}}
 1&&1&&1 \\
 &0&&3&
\end{array}\right)
 &
[4\,1;1]
 \\[-0.5pt]
23
&
-
 &
(1^5;3\,2^2\,1)
 &
 56
 &
\left(\begin{array}{@{}c@{}c@{}c@{}@{}c@{}c@{}c@{}c@{}}
 0&&2&&2 \\
 &1&&1&
\end{array}\right)
 &
[3;2\,1]
 \\[-0.5pt]
24
&
-
 &
(2^2\,1;3\,2^2\,1)
 &
 60
 &
\left(\begin{array}{@{}c@{}c@{}c@{}}
 2&&2 \\
 &2&
\end{array}\right)
 &
[\emp;2^3]
\end{array}
}\non
\end{equation}

\end{document}